\documentclass[11pt]{article}
\usepackage{graphicx}
\usepackage[font=footnotesize]{caption}

\usepackage{natbib}
\setcitestyle{authoryear,round,notesep={; },aysep={},yysep={;}}

\usepackage{amsmath,amssymb,amsfonts,amsthm,mathrsfs}
\usepackage[dvipsnames]{xcolor}
\usepackage{hyperref}
\usepackage{tabularx}
\usepackage{booktabs}
\usepackage{textcomp}
\usepackage{array}
\usepackage{placeins}

\newtheorem{theorem}{Theorem}

\numberwithin{equation}{section}
\numberwithin{table}{section}
\numberwithin{theorem}{section}
\numberwithin{lemma}{section}
\numberwithin{definition}{section}
\numberwithin{example}{section}

\makeatletter
\newcommand{\removeperiod}{\@ifnextchar.{\@gobble}\relax}
\makeatother
\hypersetup {
    colorlinks = true,
    linkcolor = blue,
    filecolor = magenta,
    urlcolor = NavyBlue,
    citecolor = NavyBlue
}

\begin{document}
\title{\LARGE \bf
Modelling Toxin-Induced Intraspecific Personality Variation in Predator-Prey Dynamics}
\author{Jingyi Jin
\thanks{Department of Mathematical and Statistical Sciences \& Interdisciplinary Lab for Mathematical Ecology and Epidemiology, University of Alberta, Edmonton, Alberta T6G 2G1, Canada}
\and Tian Xu Wang \footnotemark[1]
\and Hao Wang \thanks{Corresponding author: \href{mailto:hao8@ualberta.ca}{\texttt{hao8@ualberta.ca}}. Department of Mathematical and Statistical Sciences \& Interdisciplinary Lab for Mathematical Ecology and Epidemiology, University of Alberta, Edmonton, Alberta T6G 2G1, Canada}}

\maketitle

\begin{abstract}
Personality traits, such as boldness and shyness, can affect interspecific interactions through differences in individual behaviour. Environmental toxins may also alter these traits. However, few predator-prey models have considered animal personality. In this study, we extend a toxin-dependent predator-prey model by incorporating personality variation in both prey and predators. Prey personality affects birth and growth, whereas predator personality affects the predator functional response. The model is parameterized for rainbow trout and aquatic
insect prey exposed to methylmercury. Numerical results show that prey and predator personality have different effects on population dynamics. Prey personality tends to reduce biomass fluctuations and promote stable coexistence. In contrast, predator personality can increase biomass fluctuations and lead to persistent population oscillations. These effects are mainly observed at low and intermediate toxin concentrations that are environmentally relevant to freshwater ecosystems. Using a quasi-steady-state approximation, we reduce the model to a two-dimensional system and study its positivity, boundedness, and the existence of equilibria. The equilibrium structure depends jointly on the toxin level and the prey and predator personality contrasts. Based on this existence classification, we numerically
investigate the local stability of all feasible equilibria and
identify parameter regions corresponding to locally stable extinction, prey-only persistence, coexistence, and bistability. Overall, toxin-induced personality variation can qualitatively change predator-prey dynamics.
\end{abstract}

{\bf Keywords} {\small Environmental toxins, personality traits, predator-prey dynamics}


{\bf MSC Classification} {\small 92C15, 34A12, 34D20}

\section{Introduction}

Animals often exhibit consistent individual differences in behaviour, commonly referred to as animal personality \citep{jolle2019}. Here, we focus on activity-related behavioural differences associated with the boldness-shyness continuum. Activity level is an important dimension of animal personality and reflects consistent differences among individuals in how actively they move, forage, and feed \citep{reale2007}. Bolder individuals are often more active and exploratory and may exhibit greater foraging activity and food intake, whereas shyer individuals tend to be less active \citep{biro2008animal}. For example, in three-spined sticklebacks (\textit{Gasterosteus aculeatus}), bolder individuals consumed considerably more food than shyer individuals, even after differences in body size were taken into account \citep{jolles2016}. These behavioural differences can subsequently influence growth and reproductive performance. In the freshwater crayfish \textit{Cherax destructor}, individuals that were consistently bolder were also more voracious and maintained higher growth rates over time than shyer individuals \citep{biro2014individual}. Personality differences can also have consequences for reproductive success \citep{smith2008fitness}. For example, in largemouth bass (\textit{Micropterus salmoides}), bolder adult males achieved greater reproductive success than shyer males across three reproductive seasons under semi-natural conditions \citep{ballew2017fitness}.


Although these personality traits are generally consistent, they can
also show substantial plasticity and be shaped by environmental
stressors \citep{Bell2009,Dingemanse2010}. Environmental stressors may therefore
shift average behavioural tendencies and alter the magnitude and
structure of among-individual variation
\citep{DingemanseWolf2013,GrunstGrunst2024}. Toxic contaminants are an
important class of environmental stressors and are widely present in
aquatic ecosystems. Contaminant residues have been detected in aquatic
animals across the Pacific region \citep{Nguyen2024residal}, while high
levels of toxic metals have been reported in soils, plants, animals,
and water bodies in South Africa \citep{Okereafor2020toxic}. Metals,
pesticides, pharmaceuticals, and other organic compounds can alter
personality-related traits and affect foraging,
reproductive behaviour, and activity levels
\citep{MontiglioRoyaute2014,Saaristo2018,Jacquin2020,Bertram2022}.
For example, golden shiners exposed to dietary mercury showed increased
activity and reduced shoal cohesion, whereas
honeybee colonies exposed to syrups containing Cry1Ab protoxin showed
reduced foraging activity
\citep{WebberHaines2003,RamirezRomero2005}. Despite the ecological
importance of behavioural variation, we still have a limited
understanding of how pollution-induced behavioural changes affect
population and community dynamics.

Mathematical models provide a way to investigate these population-level effects. Several toxin-dependent population models have been developed to investigate how contaminants affect ecological interactions
\citep{hallam1983effects1,DeLuna1987model,freedman1991models,
spromberg2006relating,erickson2014daphnia}. For example,
\citet{huang2015impact} modelled methylmercury transfer from prey to
predator and described toxicity through body burden, while
\citet{wang2024stoichiometric} extended this approach by incorporating
stoichiometric effects. However, these models did not consider
toxin-induced changes in the bold-shy composition of prey and predator
populations. A recent spatial model included fixed bold and shy
phenotypes, but the proportions of these phenotypes were not toxin
dependent \citep{WANGExistence}. This is an important limitation because
environmentally induced intraspecific variation can alter interspecific
interactions and broader ecological processes
\citep{TuomainenCandolin2011,Bolnick2011,DesRoches2018}. Therefore, the effects of toxin-induced changes in the bold-shy
composition of both prey and predators on predator-prey dynamics
have not yet been addressed.

In this work, we address the gap identified above by incorporating
toxin-induced changes in personality composition into a
predator-prey model. Both prey and predator populations consist of
bold and shy individuals, and the proportions of these behavioural
types depend on toxin body burden. We use numerical simulations to
study the full model and a quasi-steady-state approximation to derive
a reduced system for mathematical analysis. Our main contributions are as follows:

\begin{enumerate}
    \item \textbf{Model formulation.}
    We incorporate toxin-dependent bold-shy composition into both
    prey and predator populations. Prey personality affects prey birth
    and growth, whereas predator personality affects the predator
    functional response. We parameterize the model for a system
    consisting of rainbow trout (\textit{Oncorhynchus mykiss}) and
    aquatic insect prey exposed to methylmercury.

    \item \textbf{Effects of prey and predator personality.}
    Through numerical simulations, we show that prey and predator
    personality have different effects on population dynamics. Prey
    personality tends to reduce population fluctuations and promote
    stable coexistence, whereas predator personality can increase
    population fluctuations and lead to persistent population
    oscillations.

    \item \textbf{Reduction and mathematical properties.}
   Using a quasi-steady-state approximation, we reduce the original
    model to a two-dimensional system. We establish the positivity and
    boundedness of the reduced system and identify a positively
    invariant region. 
    \item \textbf{Structure and local stability of equilibria.} We derive conditions for the existence of extinction, prey-only, and positive coexistence equilibria. The equilibrium structure is jointly determined by the toxin level $T$, the prey personality contrast $\rho_\alpha$, and the predator personality contrast $\rho_\mu$. The toxin thresholds determine the possible candidate roots for coexistence. For a fixed toxin level, $\rho_\alpha$ determines whether the prey can persist, while $\rho_\mu$ determines whether the system has zero, one, or two positive coexistence equilibria. Based on this existence classification, we numerically investigate the local stability of all biologically feasible equilibria and identify parameter regions corresponding to locally stable extinction, prey-only persistence, coexistence, and bistability.

\end{enumerate}

The rest of this paper is organized as follows. In Section 2, we develop a predator-prey model that incorporates toxin-induced intraspecific personality variation. In Section 3, we connect the model to experimental data via model parameterization. In Section 4, we numerically study the effects of personality and toxins on rainbow trout and aquatic insect prey. In Section 5, we reduce the model using a quasi-steady-state approximation, analyze the existence of equilibria, and numerically investigate their local stability. Section 6 provides the discussion.

\section{Model Formulation}
\label{sec: models}

Building on the toxin body-burden framework of \citet{huang2015impact}, we propose the following predator-prey model that incorporates toxin-induced intraspecific personality variation:
\begin{equation}
\label{eq: modified model}
\begin{aligned}
\underset{\text{prey biomass}}{\underbrace{\frac{\mathrm{d}x}{\mathrm{d}t}}}
&=
\underset{\text{birth and growth}}{\underbrace{b(u,x,\gamma_1(u))x}}
-
\underset{\text{mortality}}{\underbrace{d_1(u)x}}
-
\underset{\text{predation loss}}{\underbrace{p(x,\gamma_2(v))y}}
\\
\underset{\text{predator biomass}}{\underbrace{\frac{\mathrm{d}y}{\mathrm{d}t}}}
&=
\underset{\text{growth from predation}}{\underbrace{e(v)p(x,\gamma_2(v))y}}
-
\underset{\text{mortality}}{\underbrace{d_2(v)y}}
\\
\underset{\text{toxin in prey}}{\underbrace{\frac{\mathrm{d}U}{\mathrm{d}t}}}
&=
\underset{\text{uptake}}{\underbrace{a_1Tx}}
-
\underset{\text{depuration}}{\underbrace{\sigma_1U}}
-
\underset{\text{loss due to mortality}}{\underbrace{d_1(u)U}}
-
\underset{\text{predation loss}}{\underbrace{p(x,\gamma_2(v))yu}}
\\
\underset{\text{toxin in predator}}{\underbrace{\frac{\mathrm{d}V}{\mathrm{d}t}}}
&=
\underset{\text{uptake}}{\underbrace{a_2Ty}}
-
\underset{\text{depuration}}{\underbrace{\sigma_2V}}
-
\underset{\text{loss due to mortality}}{\underbrace{d_2(v)V}}
+
\underset{\text{gain from predation}}{\underbrace{p(x,\gamma_2(v))yu}}
\\
u &= U/x,
\qquad
v = V/y.
\end{aligned}
\end{equation}

Here, \(x(t)\) and \(y(t)\) denote the prey and predator biomass densities, respectively, measured in \(\mathrm{g}/\mathrm{L}\), while \(U(t)\) and \(V(t)\) denote the corresponding total toxin concentrations, measured in \(\mu\mathrm{g}/\mathrm{L}\). The body burdens \(u(t)=U(t)/x(t)\) and \(v(t)=V(t)/y(t)\) represent toxin concentration per unit biomass, measured in \(\mu\mathrm{g}/\mathrm{g}\). The functions \(d_1(u)\) and \(d_2(v)\) denote body-burden-dependent mortality rates of prey and predators, respectively, and \(e(v)\) is the efficiency with which consumed prey biomass is converted into predator biomass. The parameters \(a_1\) and \(a_2\) are environmental toxin-uptake rates, \(T\) is the environmental toxin concentration, and \(\sigma_1\) and \(\sigma_2\) are depuration rates.

Empirical evidence shows that contaminant exposure can alter activity and the expression of personality traits. For example, atrazine exposure reduced activity in initially bold \textit{Faxonius virilis} crayfish, whereas comparatively little behavioural change was observed in shy individuals \citep{steele2019}. Similarly, chlorpyrifos exposure reduced swimming activity and affected foraging-related behaviour in California killifish, \textit{Fundulus parvipinnis} \citep{renick2016}. Multigenerational exposure to a mixture of urban contaminants has also been shown to alter foraging and reproductive behaviour in fathead minnows, \textit{Pimephales promelas} \citep{swank2021}.

To incorporate the effects of toxin exposure on intraspecific variation in activity-related behaviour along the boldness-shyness continuum, we assume that the behavioural composition of each population depends on toxin body burden. Specifically, \(\gamma_1(u)\) and \(\gamma_2(v)\) denote the proportions of bold, more active individuals within the prey and predator populations, respectively. Changes in personality composition can, in turn, affect demographic and trophic processes. Bolder prey may spend more time moving and foraging and consequently obtain greater food intake \citep{jolles2016}, providing additional resources for somatic growth and reproduction. Consistent behavioural differences have accordingly been associated with variation in growth rates \citep{biro2014individual}. Prey reproduction and growth are therefore assumed to depend on the proportion of bolder individuals through \(b(u,x,\gamma_1(u))\). Similarly, bolder, more active predators may forage more actively and consequently achieve higher prey consumption rates. For example, bolder three-spined sticklebacks consumed a greater proportion of prey than shyer individuals \citep{Ioannou2008}, while higher activity levels in mud crabs increased the magnitude of their functional response to mussel prey \citep{toscano2014}. We therefore represent this behavioural effect on predation through the personality-dependent functional response \(p(x,\gamma_2(v))\).  Thus, toxin exposure can influence predator-prey dynamics both directly, through physiological effects, and indirectly, by shifting the bold-shy composition within each species.

Model \eqref{eq: modified model} can be reduced using the body-burden definitions. For brevity, primes denote differentiation with respect to time. Using \(u=U/x\) and substituting the equations for \(U'\) and \(x'\) from \eqref{eq: modified model}, we obtain
\begin{equation}
\label{eq:reduceu}
\begin{aligned}
\frac{\mathrm{d}u}{\mathrm{d}t}
&=
\frac{U'}{x}-\frac{U}{x}\frac{x'}{x}
=
a_1T-\sigma_1u-b(u,x,\gamma_1(u))u .
\end{aligned}
\end{equation}
Similarly, from \(v=V/y\), and substituting the equations for \(V'\) and \(y'\), we have
\begin{equation}
\label{eq:reducev}
\begin{aligned}
\frac{\mathrm{d}v}{\mathrm{d}t}
&=
\frac{V'}{y}-\frac{V}{y}\frac{y'}{y}
=
a_2T-\sigma_2v+p(x,\gamma_2(v))(u-e(v)v).
\end{aligned}
\end{equation}
Thus, system \eqref{eq: modified model} becomes
\begin{equation}
\label{eq:reduced-model}
\begin{aligned}
\frac{\mathrm{d}x}{\mathrm{d}t}
&= b(u,x,\gamma_1(u))x-d_1(u)x-p(x,\gamma_2(v))y, \\
\frac{\mathrm{d}y}{\mathrm{d}t}
&= e(v)p(x,\gamma_2(v))y-d_2(v)y, \\
\frac{\mathrm{d}u}{\mathrm{d}t}
&= a_1T-\sigma_1u-b(u,x,\gamma_1(u))u, \\
\frac{\mathrm{d}v}{\mathrm{d}t}
&= a_2T-\sigma_2v+p(x,\gamma_2(v))(u-e(v)v).
\end{aligned}
\end{equation}

\subsection{Specific functional forms}

Following \citet{huang2015impact}, we take
\begin{equation}
d_1(u)=k_1u+m_1,\qquad
d_2(v)=k_2v+m_2,\qquad
e(v)=\beta_1\max\{0,1-\beta_2v\},
\end{equation}
where \(k_i,m_i>0\), \(0<\beta_1\leq1\), and \(\beta_2>0\).

\subsubsection{The proportion of bold prey individuals \(\gamma_1(u)\)}

Let \(\gamma_1(u)\) and \(\gamma_2(v)\) denote the proportions of bold prey and bold predators, respectively. We model these proportions as functions of toxin body burden. By definition,
$0\leq \gamma_1(u)\leq 1$, $0\leq \gamma_2(v)\leq 1$,
so \(1-\gamma_1(u)\) and \(1-\gamma_2(v)\) are the corresponding proportions of shy individuals.

Following \citet{stitt2006evaluation}, the fitted proportion of
individuals drifting within 10 minutes is used as an operational
measure of \(\gamma_1(u)\). Details of the behavioural classification,
body-burden conversion, and fitting procedure are provided in
Appendix~\ref{app:behavioural-classification}.
The four drift proportions used in the fit are
$0.07, 0.23, 0.43, 0.46$,
corresponding to aqueous copper concentrations of
\(0.5\), \(3\), \(78\), and \(229\ \mu\mathrm{g/L}\), respectively.

Because \(\gamma_1(u)\) is a proportion, it must remain between \(0\) and \(1\). For the relevant range of \(u\), the fitted polynomial is less than \(1\), so we truncate it from below and define
\begin{equation}
\gamma_1(u)
=
\max\left\{
0,\,
-1.0895u^2+1.3155u+0.1397
\right\}.
\end{equation}
This ensures that \(\gamma_1(u)\) remains biologically meaningful.

\subsubsection{The proportion of bold predator individuals \(\gamma_2(v)\)}

To estimate \(\gamma_2(v)\), we use the adult swimming responses from
Experiment 4 of \citet{Glazer2021gammav}. Individuals with a total
distance travelled above the control-group median are classified as
bold. Because the study did not report adult
body burdens, the developmental MeHg treatments are mapped to the model
variable \(v\) using a linear scaling proxy. The classification, the
origin and interpretation of this proxy, and the fitting procedure are
described in Appendix~\ref{app:behavioural-classification}. The fitted
function is
\begin{equation}
\gamma_2(v)
=
\min\left\{
1,\,
81740.5109v^2-358.1220v+0.4872
\right\}.
\end{equation}
The upper truncation ensures that \(\gamma_2(v)\) remains a proportion.

\subsubsection{The prey birth and growth rate \(b(u,x,\gamma_1(u))\)}

In \citet{huang2015impact}, the prey birth and growth rate is given by
\begin{equation}
b(u,x)
=
\alpha_1
\frac{\max\{0,1-\alpha_2u\}}{1+\alpha_3x},
\end{equation}
where the body burden \(u\) reduces prey reproduction and growth, while the half-saturation form captures density-dependent competition and crowding effects.

Because behavioural phenotype can influence growth, fecundity, and reproductive success \citep{biro2008animal,smith2008fitness,biro2014individual,ballew2017fitness}, we consider prey birth and growth to depend on the bold-shy composition of the prey population. With \(\gamma_1(u)\) denoting the proportion of bold prey individuals, define
\begin{equation}
B(u)=\alpha_{11}\gamma_1(u)+\alpha_{12}\left[1-\gamma_1(u)\right],
\end{equation}
where \(\alpha_{11}\) and \(\alpha_{12}\) are the maximum birth and growth rates of bold and shy prey individuals, respectively. We then define
\begin{equation}
\begin{aligned}
b(u,x,\gamma_1(u))
&=B(u)
\frac{\max\{0,1-\alpha_2u\}}{1+\alpha_3x}\\
&=
\gamma_1(u)\alpha_{11}
\frac{\max\{0,1-\alpha_2u\}}{1+\alpha_3x}
+
\left(1-\gamma_1(u)\right)\alpha_{12}
\frac{\max\{0,1-\alpha_2u\}}{1+\alpha_3x}.
\end{aligned}
\end{equation}
The first term represents the contribution from bold prey, whereas the second represents the contribution from shy prey. In particular, if \(\gamma_1(u)=1\), then all prey individuals are classified as bold and
\[
b(u,x,\gamma_1(u))
=
\alpha_{11}
\frac{\max\{0,1-\alpha_2u\}}{1+\alpha_3x}.
\]

\subsubsection{The predation function \(p(x,\gamma_2(v))\)}

We also incorporate the intraspecific bold-shy composition of predators into the predation process. With \(\gamma_2(v)\) denoting the proportion of bold predators, we define
\begin{equation}
p(x,\gamma_2(v))
=
\gamma_2(v)p_1(x)
+
\left(1-\gamma_2(v)\right)p_2(x),
\end{equation}
where \(p_1(x)\) and \(p_2(x)\) are the functional responses of bold and shy predators, respectively. Assume that both predator types follow Holling Type II functional responses:
\begin{equation}
p_1(x)=\frac{\mu_1x}{1+h\mu_1x},
\qquad
p_2(x)=\frac{\mu_2x}{1+h\mu_2x},
\end{equation}
where \(\mu_1\) and \(\mu_2\) are the encounter rates of bold and shy predators with prey, respectively, and \(h\) is the handling time. We assume that bold predators have a higher encounter rate with prey than shy predators \citep{Ioannou2008}; hence, we take
\[
\mu_1>\mu_2.
\]
To simplify notation, let
\(
\xi=\frac{1}{h},
\eta_1=\frac{1}{\mu_1h},
\eta_2=\frac{1}{\mu_2h}.
\)
Then
\begin{equation}
p_1(x)=\frac{\xi x}{\eta_1+x},
\qquad
p_2(x)=\frac{\xi x}{\eta_2+x}.
\end{equation}
Here, \(\eta_i\) is the half-saturation constant of \(p_i(x)\). Since \(\mu_1>\mu_2\), we have \(\eta_1<\eta_2\). Therefore, the personality-dependent predation function is
\begin{equation}
p(x,\gamma_2(v))
=
\gamma_2(v)\frac{\xi x}{\eta_1+x}
+
\left(1-\gamma_2(v)\right)\frac{\xi x}{\eta_2+x}.
\end{equation}
Using these functions, the final model is
\begin{equation}
\label{eq: final model}
\begin{aligned}
\frac{\mathrm{d}x}{\mathrm{d}t}
&=
B(u)\frac{\max\{0,1-\alpha_2u\}x}{1+\alpha_3x}
-(k_1u+m_1)x
-p(x,\gamma_2(v))y,\\
\frac{\mathrm{d}y}{\mathrm{d}t}
&=
\beta_1\max\{0,1-\beta_2v\}p(x,\gamma_2(v))y
-(k_2v+m_2)y,\\
\frac{\mathrm{d}u}{\mathrm{d}t}
&=
a_1T-\sigma_1u
-B(u)\frac{\max\{0,1-\alpha_2u\}}{1+\alpha_3x}u,\\
\frac{\mathrm{d}v}{\mathrm{d}t}
&=
a_2T-\sigma_2v
+p(x,\gamma_2(v))
\left(u-\beta_1\max\{0,1-\beta_2v\}v\right).
\end{aligned}
\end{equation}

\section{Parameterization}

In this section, we parameterize model~\eqref{eq: final model} for a
predator-prey system consisting of rainbow trout and aquatic insect
prey. Most parameters are taken from
\citet{huang2015impact}, whereas the personality-related parameters
and functions are described below.

\paragraph*{The encounter rate}

The half-saturation constant estimated by \citet{huang2015impact} is
\(\eta=0.000455\ \mathrm{g/L}\), and the per capita feeding-rate parameter is
\(\xi=0.047\ \mathrm{day}^{-1}\). These values yield an effective
predator-prey encounter rate of
\[
\mu_1=\frac{\xi}{\eta}\approx 103\ \mathrm{L\,g^{-1}\,day^{-1}}.
\]

For the numerical simulations, we take the encounter rate $\mu_1$ given above as the value for bold predators and assume a moderate \(10\%\) personality-induced difference between bold and shy individuals relative to a standard encounter rate. Accordingly, we define the baseline encounter rate as
\[ \bar{\mu}=93\ \mathrm{L\,g^{-1}\,day^{-1}}. \]
Then the encounter rate for shy predators is assumed to be \(10\%\) lower than the baseline value, giving
\[ \mu_2=0.9\bar{\mu}=83.7\ \mathrm{L\,g^{-1}\,day^{-1}}. \]
Here, \(\mu_1\) and \(\mu_2\) denote the encounter rates associated with the bold and shy predator personality types, respectively.

\paragraph*{Maximum birth and growth rate of aquatic insects}

The maximum birth and growth rate of aquatic insects,
\(\alpha_1 = 4/365\ \mathrm{day}^{-1}\),
estimated by \citet{huang2015impact}, represents the average maximum birth and growth rate of the prey population.

Because bold individuals may have higher growth rates or reproductive success than shy individuals \citep{biro2014individual,ballew2017fitness}, we assume a moderate \(20\%\) contrast around the mean maximum birth and growth rate. Accordingly, we set
\[
\alpha_{11} = 1.2\alpha_1 = \frac{4.8}{365}\ \mathrm{day}^{-1},
\qquad
\alpha_{12} = 0.8\alpha_1 = \frac{3.2}{365}\ \mathrm{day}^{-1},
\]
where \(\alpha_{11}\) and \(\alpha_{12}\) denote the maximum birth and growth rates of bold and shy aquatic insects, respectively.

The parameter estimates for rainbow trout and aquatic insects are listed in Table~\ref{Table: Parameters}.

\setlength{\arrayrulewidth}{0.35mm} 
\begin{table}[h]
\footnotesize
\centering
\caption{\footnotesize Parameters in the system \eqref{eq: final model}}
\label{Table: Parameters}
\begin{tabular}{l l l l}
\hline
Symbols & Definitions & Values used & Source \\
\hline
$\alpha_{11}$ & Maximum birth and growth rate of bold aquatic insects & 4.8/365 $\text{day}^{-1}$ & $\vartriangle$\\
$\alpha_{12}$ & Maximum birth and growth rate of shy aquatic insects & 3.2/365 $\text{day}^{-1}$ & $\vartriangle$\\
$\alpha_2$ & Effect of toxin on the growth of aquatic insects & 0.022 g/$\mu$g & $\bullet$\\
$\alpha_3$ & Crowding effect of aquatic insects & 1330 L/g & $\bullet$\\
$k_1$ & Effect of toxin on the mortality of aquatic insects & 0.0021 g/$\mu$g/day & $\bullet$\\
$m_1$ & Natural mortality rate of aquatic insects & 1/365 $\text{day}^{-1}$ & $\bullet$\\
$\xi$ & Per capita feeding rate & 0.047 $\text{day}^{-1}$ & $\bullet$\\
$\beta_1$ & Reproduction efficiency of rainbow trout & 0.1 & $\bullet$\\
$\beta_2$ & Effect of toxin on the reproduction of rainbow
trout & 33.41 g/$\mu$g & $\bullet$\\
$k_2$ & Effect of toxin on the mortality of rainbow trout & 0.00398 g/$\mu$g/day & $\bullet$\\
$m_2$ & Natural mortality rate of rainbow trout & 0.00057 $\text{day}^{-1}$ & $\bullet$\\
$a_1$ & Uptake coefficient for aquatic insects & 0.55 L/g/day & $\bullet$\\
$\sigma_1$ & Depuration rate for aquatic insects & 0.12 $\text{day}^{-1}$ & $\bullet$\\
$a_2$ & Uptake coefficient for rainbow trout & 0.1733 L/g/day & $\bullet$\\
$\sigma_2$ & Depuration rate for rainbow trout & 0.0062 $\text{day}^{-1}$ & $\bullet$\\
$\bar{\mu}$ & Baseline predator-prey encounter rate & 93 L/g/day & $\vartriangle$\\
$\mu_1$ & Encounter rate for bold predators & 102.3 L/g/day & $\vartriangle$\\
$\mu_2$ & Encounter rate for shy predators & 83.7 L/g/day & $\vartriangle$\\
\hline
\end{tabular}
\caption*{\footnotesize 
    Note: $\bullet$ Taken from \citet{huang2015impact}. \(\triangle\) Derived from the baseline estimate.}
\end{table}

\section{Numerical results and personality effects}
\label{sec4}

In this section, we examine how personality traits affect the long-term dynamics of the system. The parameter values listed in Table~\ref{Table: Parameters} are used to numerically solve the personality- and toxin-dependent predator-prey model \eqref{eq: final model}.

To better understand how toxin-induced intraspecific personality changes influence predator–prey dynamics, we compare four scenarios that differ in whether personality variation is present in the prey, the predator, or both species: (i) no personality variation, (ii) personality variation in the prey only, (iii) personality variation in the predator only, and (iv) personality variation in both prey and predator. We then vary the strength of personality-induced
behavioral differences to investigate their effects on oscillation
amplitude and system stability.


The main oscillatory changes occur for
\(0\le T\le 0.001\,\mu\mathrm{g/L}\). This concentration range
corresponds to \(0\)-\(1\,\mathrm{ng/L}\), which is environmentally
relevant for freshwater systems. Methylmercury concentrations in natural
waters are usually below \(1\,\mathrm{ng/L}\) \citep{CCME2003Hg}.
Field studies have also reported lake-water MeHg concentrations in this
range, including \(0.01\)-\(0.34\,\mathrm{ng/L}\) in lakes of Canada's
Oil Sands Region \citep{Emmerton2018Hg}; a global lake-water dataset
reported a mean MeHg concentration of \(0.14\,\mathrm{ng/L}\)
\citep{Lei2026MeHg}. Therefore, the following analyses mainly focus on
\(0\le T\le 0.001\,\mu\mathrm{g/L}\).

\subsection{Effects of the presence of personality in prey and predators}

In this section, we investigate how the presence or absence of personality
in prey and predators alters the toxin-dependent long-term dynamics and
stability structure of the predator-prey system. We consider four
cases: no personality, prey personality only, predator personality only,
and personality in both prey and predators. These cases are constructed from the model \eqref{eq: final model} and
differ only in the prey birth and growth term \(B(u)\) and the predator
foraging response \(p(x,\gamma_2(v))\), as summarized in
Table~\ref{tab:model_variants}. 

When personality is absent from both prey and predators
(Fig.~\ref{fig:variant_comparison}a), all individuals within each species
are represented by the same behavioural response. At low toxin
concentrations, the prey and predator biomasses exhibit persistent
oscillations. As the toxin concentration increases, the oscillations
disappear and the system approaches a stable coexistence equilibrium.
At still higher toxin concentrations, the predator becomes extinct and
the system transitions to a stable prey-only equilibrium. Thus,
increasing toxin initially suppresses the oscillatory dynamics but
eventually results in predator loss.

When personality is included only in the prey population
(Fig.~\ref{fig:variant_comparison}b), no oscillatory interval is observed
over the toxin range considered. Instead, the system remains at a stable
coexistence equilibrium at low and moderate toxin concentrations before
transitioning to a prey-only equilibrium at high toxin concentrations.
Compared with the no-personality case, the presence of prey personality
therefore suppresses the low-toxin oscillatory regime.

More complex dynamics arise when personality is included only in the
predator population (Fig.~\ref{fig:variant_comparison}c). At very low
toxin concentrations, the system exhibits oscillatory dynamics. This
initial oscillatory regime is followed by stable coexistence for
approximately
\(
0.04\times10^{-3}
\lesssim T \lesssim
0.18\times10^{-3}\,\mu\mathrm{g/L}.
\)
As the toxin concentration increases further, a second oscillatory
interval occurs approximately over
\(
0.19\times10^{-3}
\lesssim T \lesssim
0.42\times10^{-3}\,\mu\mathrm{g/L}.
\)
The system then returns to a stable coexistence equilibrium before
transitioning to a prey-only state at approximately
\(
T=0.5\times10^{-3}\,\mu\mathrm{g/L}.
\)
Thus, the presence of predator personality reorganizes the
toxin-dependent dynamics by producing two separate oscillatory
intervals. 

When personality is included in both prey and predators
(Fig.~\ref{fig:variant_comparison}d), the system remains at a stable
coexistence equilibrium at low toxin concentrations. An intermediate
oscillatory interval, with relatively small biomass fluctuations, occurs
approximately over
\(
0.18\times10^{-3}
\lesssim T \lesssim
0.30\times10^{-3}\,\mu\mathrm{g/L}.
\)
This interval is followed by another stable coexistence regime and,
at sufficiently high toxin concentrations, a stable prey-only
equilibrium. Compared with the predator-personality-only case, including
prey personality removes the very-low-toxin oscillatory interval and
reduces the extent of the intermediate oscillatory regime. This comparison
shows that differences in prey birth and growth can partly offset
the destabilizing effect of differences in predator foraging.

Overall, all four cases exhibit predator extinction when the toxin
concentration becomes sufficiently high. The principal differences
among the cases occur at low and moderate toxin concentrations, where
the presence and trophic location of personality determine whether the
system approaches a coexistence equilibrium or exhibits persistent
oscillations. These results demonstrate that the presence of personality
in prey and predators can substantially reshape the toxin-dependent
long-term dynamics of the predator-prey
system.

\begin{figure}[htbp]
    \centering

    \begin{minipage}{0.48\textwidth}
        \centering
        \includegraphics[width=\textwidth]{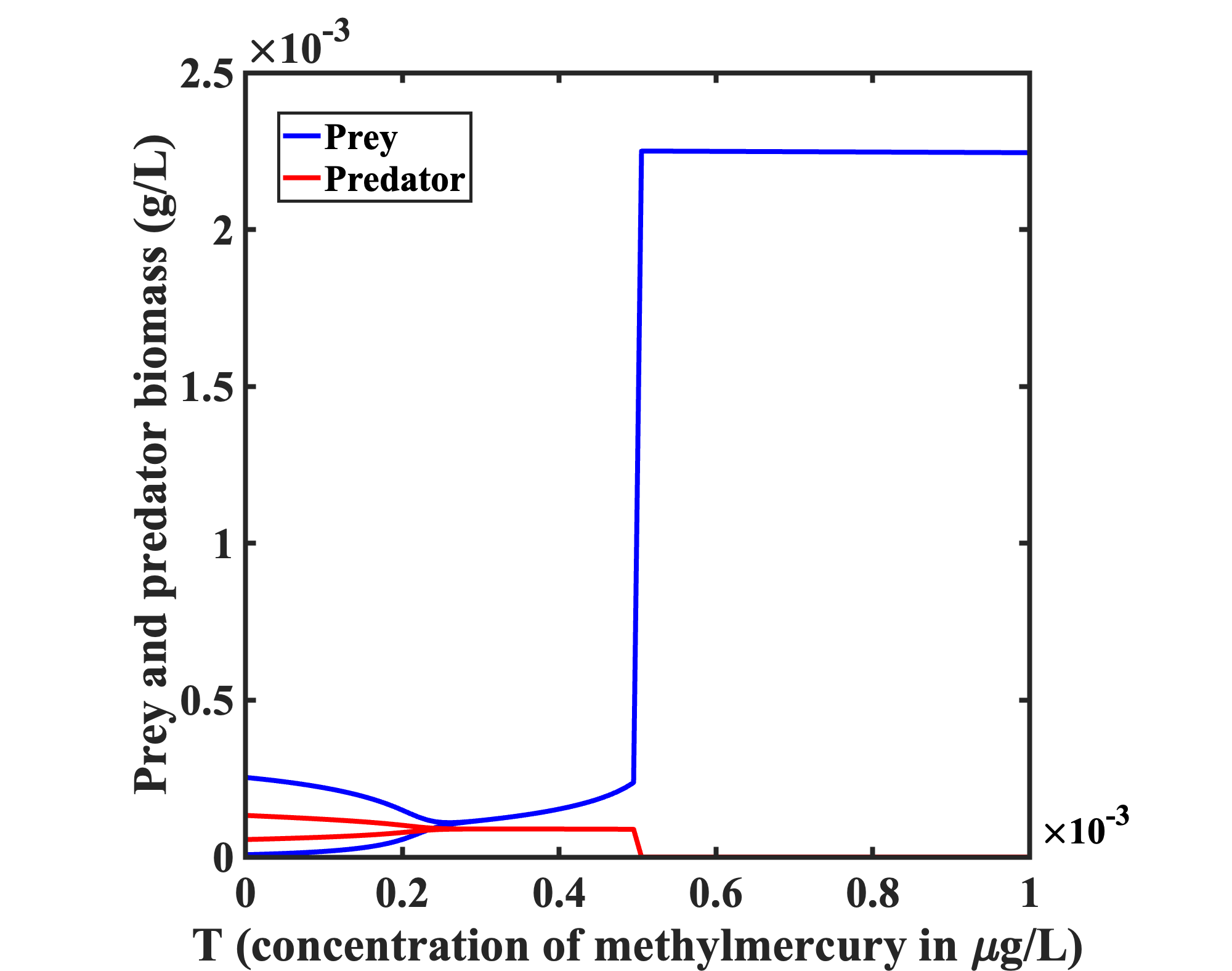}
        \caption*{(a) No personality}
    \end{minipage}
    \hfill
    \begin{minipage}{0.48\textwidth}
        \centering
        \includegraphics[width=\textwidth]{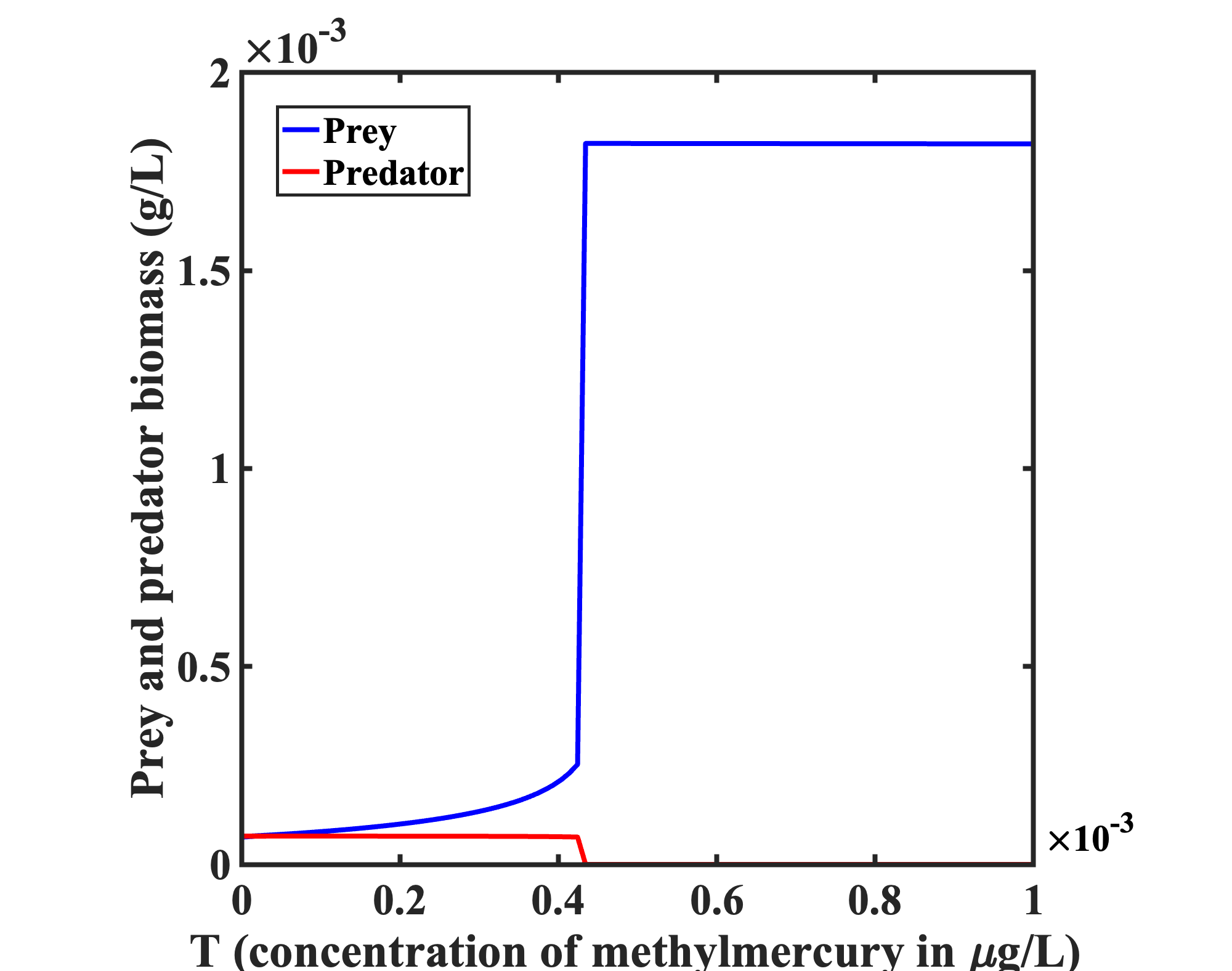}
        \caption*{(b) Prey personality only}
    \end{minipage}

    \vspace{0.3cm}

    \begin{minipage}{0.48\textwidth}
        \centering
        \includegraphics[width=\textwidth]{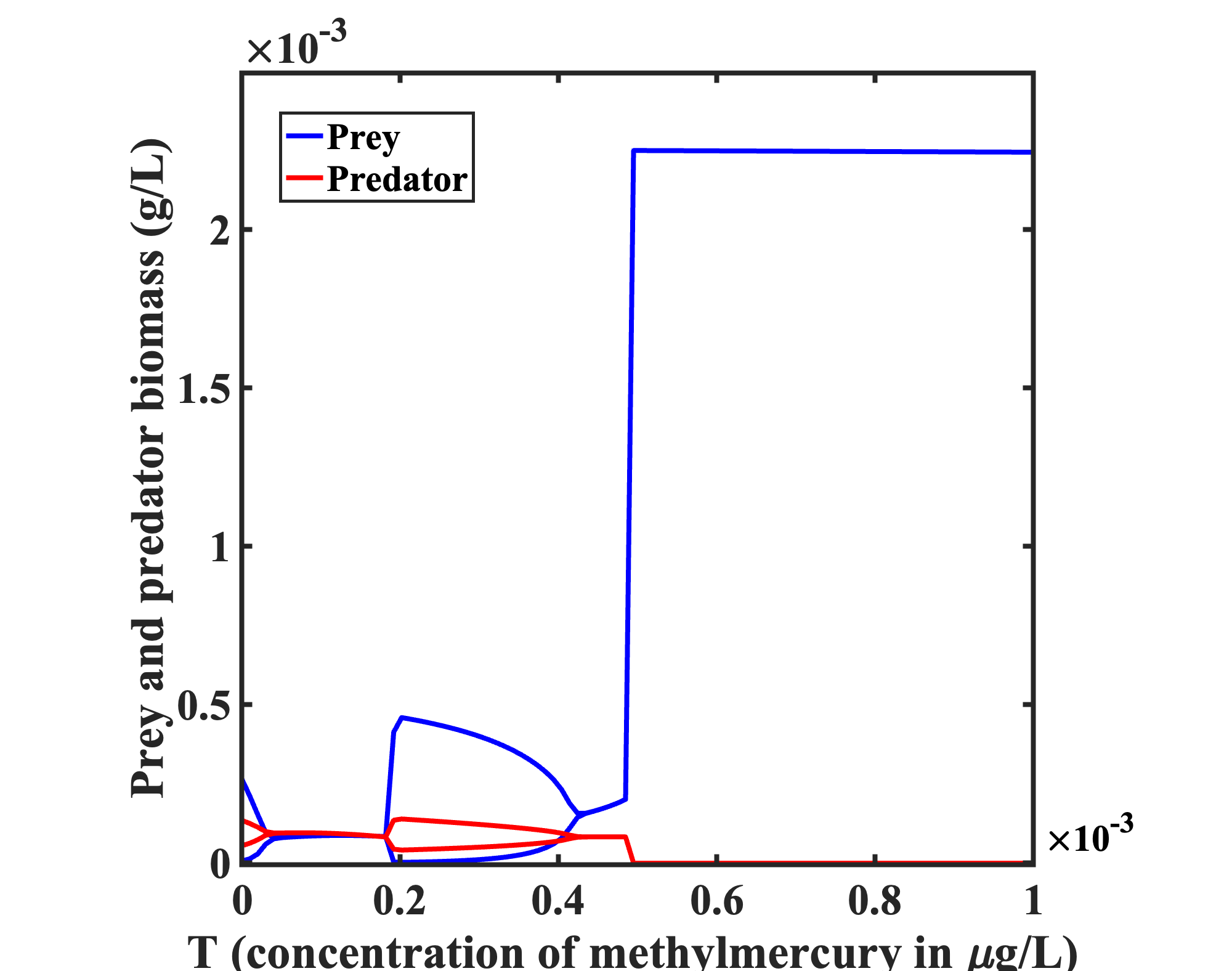}
        \caption*{(c) Predator personality only}
    \end{minipage}
    \hfill
    \begin{minipage}{0.48\textwidth}
        \centering
        \includegraphics[width=\textwidth]{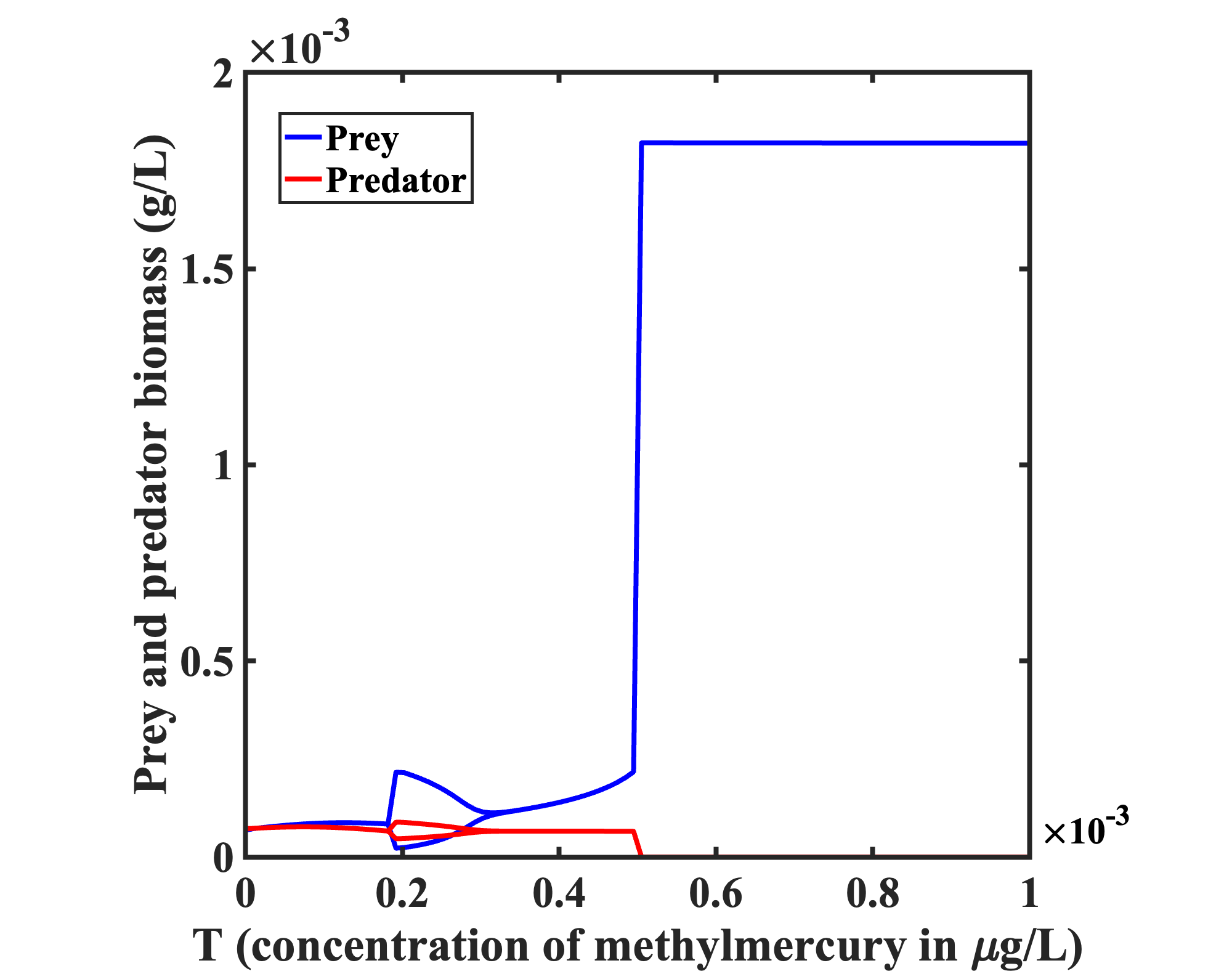}
        \caption*{(d) Prey and predator personality}
    \end{minipage}
    \caption{
    Bifurcation diagrams for the four model variants. The no-personality and prey-only models use \(\mu=93\), while the predator-only and full personality models use \(\mu_1=102.3\) and \(\mu_2=83.7\), corresponding to a \(10\%\) contrast around \(\mu=93\). All other parameters are as in Table~\ref{Table: Parameters}. Each panel shows the long-term biomass behaviour of the prey (aquatic insects) and predator (rainbow trout) as the environmental methylmercury concentration \(T\) varies from \(0\) to \(0.001\,\mu\mathrm{g/L}\).}
\label{fig:variant_comparison}
\end{figure}

\begin{table}[ht]
\centering
\caption{Model variants used to separate the effects of prey and predator personality.}
\label{tab:model_variants}
\resizebox{\textwidth}{!}{
\begin{tabular}{llll}
\hline
Panel & Model variant & Prey birth and growth term \(B(u)\) & Predator foraging response \(p(x,\gamma_2(v))\) \\
\hline
(a) & No personality &
\(\alpha_1\) &
\(\displaystyle \frac{\xi x}{x+\xi/\mu}\) \\[2mm]

(b) & Prey personality only &
\(\displaystyle \alpha_{11}\gamma_1(u)+\alpha_{12}\bigl(1-\gamma_1(u)\bigr)\) &
\(\displaystyle \frac{\xi x}{x+\xi/\mu}\) \\[2mm]

(c) & Predator personality only &
\(\alpha_1\) &
\(\displaystyle
\frac{\xi x}{x+\xi/\mu_1}\gamma_2(v)+\frac{\xi x}{x+\xi/\mu_2}\bigl(1-\gamma_2(v)\bigr)\) \\[5mm]

(d) & Prey and predator personality &
\(\displaystyle \alpha_{11}\gamma_1(u)+\alpha_{12}\bigl(1-\gamma_1(u)\bigr)\) &
\(\displaystyle
\frac{\xi x}{x+\xi/\mu_1}\gamma_2(v)+\frac{\xi x}{x+\xi/\mu_2}\bigl(1-\gamma_2(v)\bigr)\) \\
\hline
\end{tabular}
}
\end{table}

\subsection{Effects of personality-driven behavioural contrast on long-term dynamics}

\begin{figure}[htbp]
    \centering

    \begin{minipage}{0.49\textwidth}
        \centering
        \includegraphics[width=\textwidth]{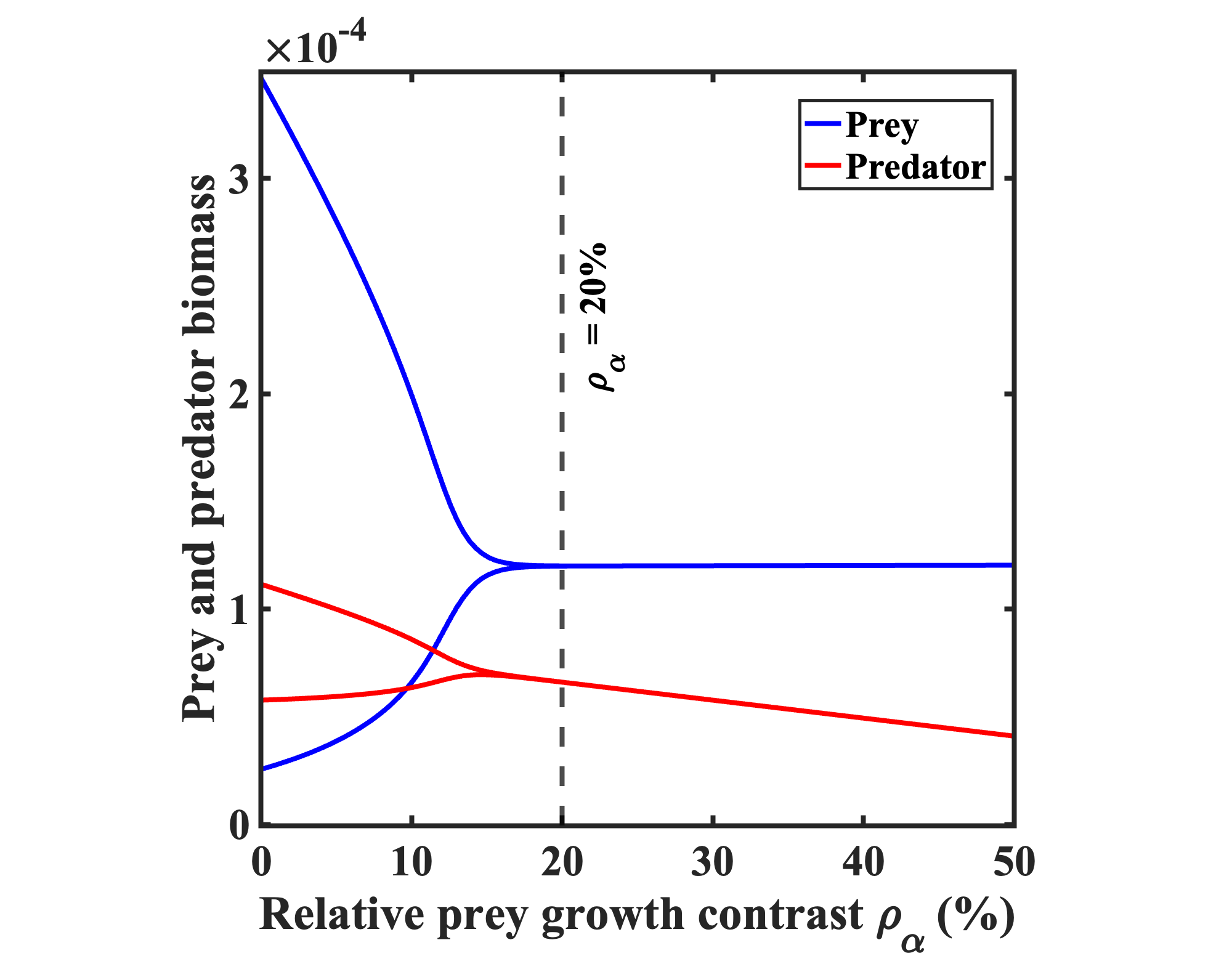}
        \caption*{(a) Prey personality contrast: biomass}
    \end{minipage}
    \hfill
    \begin{minipage}{0.49\textwidth}
        \centering
        \includegraphics[width=\textwidth]{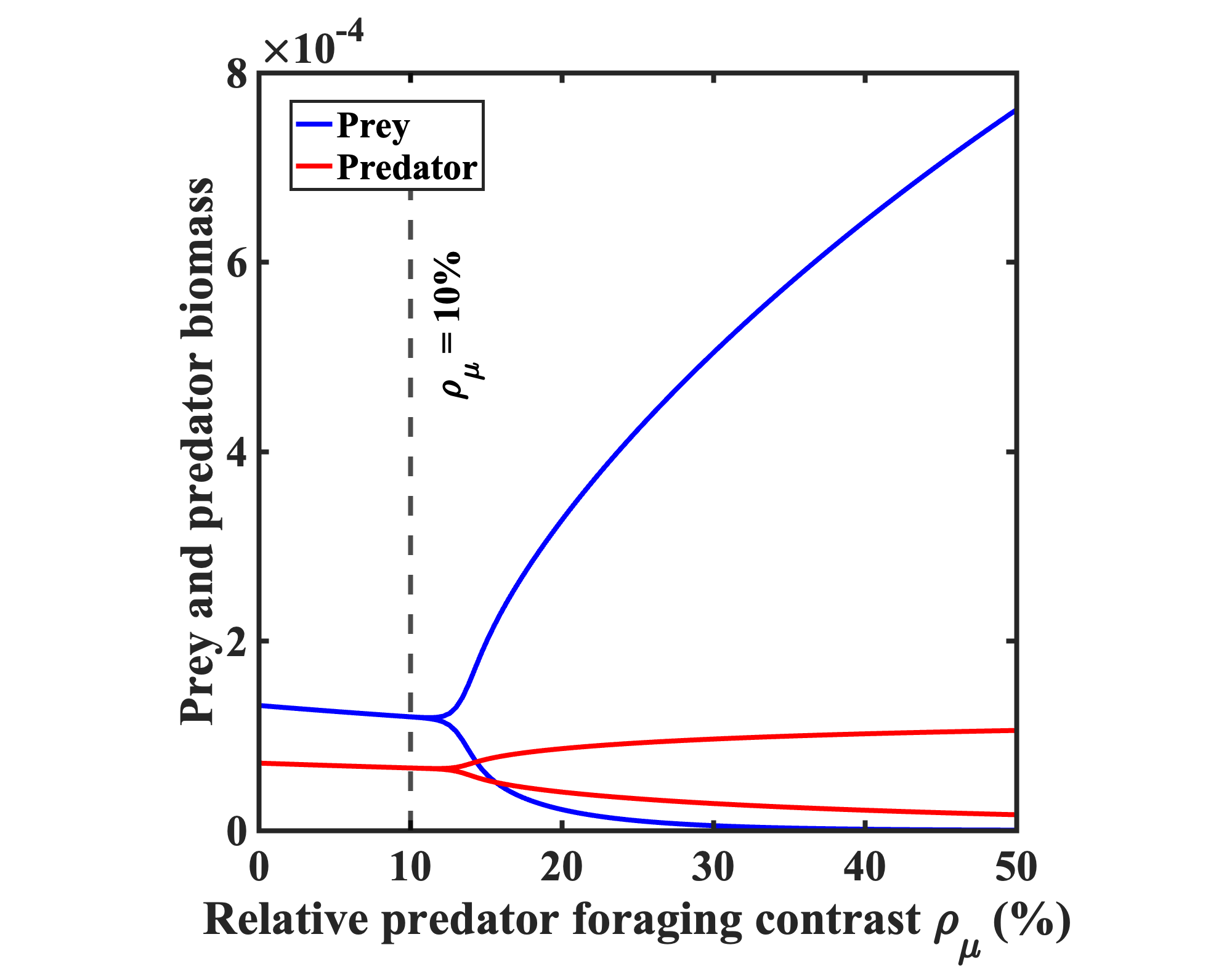}
        \caption*{(b) Predator personality contrast: biomass}
    \end{minipage}

    \caption{
    Effects of personality-driven behavioural differentiation on long-term
    biomass dynamics. Panel (a) varies the prey personality contrast
    \(\rho_\alpha\), whereas panel (b) varies the predator foraging
    contrast \(\rho_\mu\). Dashed vertical lines mark the contrast values
    used in the full personality model in Fig.~\ref{fig:variant_comparison}d.
    All simulations use \(T=0.00035\,\mu\mathrm{g/L}\).
    }
    \label{fig:contrast_effects}
\end{figure}

In this section, we examine how increasing personality-induced behavioural differences affect long-term biomass dynamics. We introduce
two relative contrast measures.

For predators, the relative personality contrast is defined as
\[
\rho_\mu
:=
\frac{\mu_1-\mu_2}{2\bar{\mu}},
\qquad
\bar{\mu}:=\frac{\mu_1+\mu_2}{2}.
\]
Thus,
\[
\mu_1=\bar{\mu}(1+\rho_\mu),
\qquad
\mu_2=\bar{\mu}(1-\rho_\mu).
\]
The parameter \(\rho_\mu\) measures the magnitude of the
personality-induced difference in predator foraging response. A larger
\(\rho_\mu\) means that personality has a stronger effect on predator
foraging behaviour.

For prey, the relative personality contrast is defined as
\[
\rho_\alpha
:=
\frac{\alpha_{11}-\alpha_{12}}{2\bar{\alpha}},
\qquad
\bar{\alpha}:=\frac{\alpha_{11}+\alpha_{12}}{2}.
\]
Thus,
\[
\alpha_{11}=\bar{\alpha}(1+\rho_\alpha),
\qquad
\alpha_{12}=\bar{\alpha}(1-\rho_\alpha).
\]
The parameter \(\rho_\alpha\) measures the magnitude of the
personality-induced difference in prey birth and growth. A larger
\(\rho_\alpha\) means that personality has a stronger effect on prey
birth and growth. Both contrasts are reported as percentages in the
figures.
To preserve positive baseline rates, we restrict both contrasts to
\begin{equation}
0\leq\rho_\alpha<1,
\qquad
0\leq\rho_\mu<1.
\label{eq:contrast-domains}
\end{equation}

The two contrast parameters have opposite effects on the long-term dynamics (Fig.~\ref{fig:contrast_effects}). As the prey personality contrast
\(\rho_\alpha\) increases, the oscillation amplitudes decrease (Fig.~\ref{fig:contrast_effects}a). For sufficiently large
\(\rho_\alpha\), the oscillations disappear and the system converges to a stable
equilibrium. Thus, when personality has a stronger effect on prey
birth and growth, the system tends to become more stable.

In contrast, increasing the predator personality contrast \(\rho_\mu\)
destabilizes the system (Fig.~\ref{fig:contrast_effects}b). At low
\(\rho_\mu\), the system converges to a stable equilibrium. As
\(\rho_\mu\) increases, persistent oscillations appear and their amplitudes
increase. Thus, when personality has a stronger effect on predator
foraging response, the system is more likely to exhibit large biomass
oscillations.

These results show that personality variation affects long-term biomass
dynamics through the magnitude of its effect on biological behaviour.
In these simulations, personality-driven differences in prey birth and
growth suppress oscillations, whereas differences in predator foraging
generate and amplify oscillations. These findings agree with the
model-variant comparison above, where prey personality
has a stabilizing effect and predator personality has a destabilizing effect.

\section{Model analysis}
In this section, we nondimensionalize the model and apply a
quasi-steady-state approximation to obtain a two-dimensional system.
We establish the positivity and boundedness of the reduced system,
derive conditions for equilibrium existence, and numerically classify
the local stability of all biologically feasible equilibria.

\subsection{Nondimensionalization}

We rescale system~\eqref{eq: final model} by introducing the following nondimensional quantities:
\begin{equation}
\begin{aligned}
&\widetilde{x}=\alpha_3x,\quad
\widetilde{y}=\frac{\alpha_3\xi}{\alpha_1}y,\quad
\widetilde{u}=\alpha_2u,\quad
\widetilde{v}=\beta_2v,\quad
\widetilde{t}=\alpha_1t,\quad
\widetilde{k}_1=\frac{k_1}{\alpha_1\alpha_2},\quad
\widetilde{k}_2=\frac{k_2}{\alpha_1\beta_2},\\
&\widetilde{m}_1=\frac{m_1}{\alpha_1},\quad
\widetilde{m}_2=\frac{m_2}{\alpha_1},\quad
\widetilde{\beta}_1=\frac{\beta_1\xi}{\alpha_1},\quad
\widetilde{\beta}_2=\frac{\beta_2\xi}{\sigma_1\alpha_2},\quad
c=\frac{\beta_2a_2}{a_1\alpha_2},\quad
\epsilon=\frac{\alpha_1}{\sigma_1},\\
&\widetilde{T}=\frac{\alpha_2a_1T}{\sigma_1},\quad
\widetilde{\sigma}_2=\frac{\sigma_2}{\sigma_1},\quad
\widetilde{\eta}_1=\alpha_3\eta_1,\quad
\widetilde{\eta}_2=\alpha_3\eta_2,\quad
\widetilde{\alpha}_{11}=\frac{\alpha_{11}}{\alpha_1},\quad
\widetilde{\alpha}_{12}=\frac{\alpha_{12}}{\alpha_1}.
\end{aligned}
\end{equation}

Dropping the tildes for simplicity, all variables and parameters below are nondimensional. Since the original dimensional parameters $\alpha_2$ and $\beta_2$ still appear in $\gamma_1$ and $\gamma_2$, we denote them by $\widehat{\alpha}_2$ and $\widehat{\beta}_2$, respectively. The nondimensional system is then given by:

\begin{equation}
\label{eq:nondimension}
\resizebox{\textwidth}{!}{ 
    $\displaystyle
\begin{aligned}
& \frac{\mathrm{d}x}{\mathrm{d}t}=
\Bigl(\Bigl\{\alpha_{11}\gamma_1(\frac{u}{\widehat{\alpha}_{2}})+\alpha_{12} \bigl[1-\gamma_1(\frac{u}{\widehat{\alpha}_{2}})\bigl]\Bigl\}\frac{\max\{0,1-u\}}{1+x}-k_1u-m_1\Bigl)x-\Bigl\{\frac{x}{x+\eta_1}\gamma_2(\frac{v}{\widehat{\beta}_{2}})+\frac{x}{x+\eta_2}\bigl[1-\gamma_2(\frac{v}{\widehat{\beta}_{2}})\bigl]\Bigl\}y\\
& \frac{dy}{dt}=\beta_1\max\{0,1-v\}\Bigl\{\frac{x}{x+\eta_1}\gamma_2(\frac{v}{\widehat{\beta}_{2}})+\frac{x}{x+\eta_2}\bigl[1-\gamma_2(\frac{v}{\widehat{\beta}_{2}})\bigl]\Bigl\}y-(k_2v+m_2)y \\
& \epsilon\frac{\mathrm{d}u}{\mathrm{d}t}=T-u-\epsilon \Bigl\{\alpha_{11}\gamma_1(\frac{u}{\widehat{\alpha}_{2}})+\alpha_{12} \bigl[1-\gamma_1(\frac{u}{\widehat{\alpha}_{2}})\bigl]\Bigl\} \frac{\max\{0,1-u\}}{1+x}u\\
& \epsilon\frac{\mathrm{d}v}{\mathrm{d}t}=cT-\sigma_2v+\Bigl(\beta_2u-\epsilon\beta_1\max\{0,1-v\}v\Bigl)\Bigl\{\frac{x}{x+\eta_1}\gamma_2(\frac{v}{\widehat{\beta}_{2}})+\frac{x}{x+\eta_2}\bigl[1-\gamma_2(\frac{v}{\widehat{\beta}_{2}})\bigl]\Bigl\}
\end{aligned}
$}
\end{equation}

\subsection{Quasi-steady system}

We simplify system~\eqref{eq:nondimension} to a two-dimensional system using the quasi-steady-state approximation. The parameter $\epsilon=\frac{\alpha_1}{\sigma_1}$
is small. Based on the parameter estimates for aquatic insects and rainbow trout in Table~\ref{Table: Parameters}, we obtain $\epsilon=0.091$. This separation of time scales motivates the formal limit $\epsilon\rightarrow0$ in system~\eqref{eq:nondimension}, which gives
\begin{equation}
\left\{
\begin{array}{l}
T-u=0,\\[1mm]
cT-\sigma_2v+\beta_2u
\left\{
\displaystyle
\frac{x}{x+\eta_1}
\gamma_2\left(\frac{v}{\widehat{\beta}_2}\right)
+
\frac{x}{x+\eta_2}
\left[
1-\gamma_2\left(\frac{v}{\widehat{\beta}_2}\right)
\right]
\right\}
=0.
\end{array}
\right.
\label{eq:qssa_equations}
\end{equation}

Here,
\begin{equation}
\gamma_2\left(\frac{v}{\widehat{\beta}_2}\right)
=
\min\left\{
1,\,
a\left(\frac{v}{\widehat{\beta}_2}\right)^2
-b\frac{v}{\widehat{\beta}_2}
+s
\right\},
\label{eq:gamma2_qssa}
\end{equation}
where
\[
a=81740.5109,\qquad
b=358.1220,\qquad
s=0.4872.
\]

We consider the following two cases.

\noindent\textbf{Case 1.}
If 
\begin{equation}
a\left(\frac{v}{\widehat{\beta}_2}\right)^2
-b\frac{v}{\widehat{\beta}_2}
+s
\leq1,
\label{eq:gamma2_below_one}
\end{equation}
then $\gamma_2\left(\frac{v}{\widehat{\beta}_2}\right)
=
a\left(\frac{v}{\widehat{\beta}_2}\right)^2
-b\frac{v}{\widehat{\beta}_2}
+s$.
Substituting this expression into \eqref{eq:qssa_equations} and multiplying the second equation by $\widehat{\beta}_2^2$, we obtain
\begin{equation}
\left\{
\begin{array}{l}
u=T,\\[1mm]
a(x)v^2+b(x)v+c(x)=0.
\end{array}
\right.
\label{eq:qssa_quadratic}
\end{equation}

\noindent where
\begin{equation}
\left\{
\begin{aligned}
a(x)
&=
a\hat\beta_2T
\left(
\frac{x}{x+\eta_1}
-\frac{x}{x+\eta_2}
\right),\\
b(x)
&=
-b\beta_2T\widehat{\beta}_2
\left(
\frac{x}{x+\eta_1}
-\frac{x}{x+\eta_2}
\right)
-\widehat{\beta}_2^2\sigma_2,\\
c(x)
&=
\widehat{\beta}_2^2
\left[
cT
+s\beta_2T
\left(
\frac{x}{x+\eta_1}
-\frac{x}{x+\eta_2}
\right)
+\beta_2T\frac{x}{x+\eta_2}
\right].
\end{aligned}
\right.
\label{eq:qssa_coefficients}
\end{equation}

Therefore,
\begin{equation}
\left\{
\begin{array}{l}
u=T,\\[1mm]
\displaystyle
v_{\pm}(x)
=
\frac{-b(x)\pm\sqrt{b^2(x)-4a(x)c(x)}}{2a(x)}.
\end{array}
\right.
\label{eq:v_quadratic_solutions}
\end{equation}
Only nonnegative real solutions satisfying \eqref{eq:gamma2_below_one} are considered.

\noindent\textbf{Case 2.}
If
\begin{equation}
a\left(\frac{v}{\widehat{\beta}_2}\right)^2
-b\frac{v}{\widehat{\beta}_2}
+s
>1,
\label{eq:gamma2_above_one}
\end{equation}
then $\gamma_2\left(\frac{v}{\widehat{\beta}_2}\right)=1$.
Equation \eqref{eq:qssa_equations} becomes
\begin{equation}
\left\{
\begin{array}{l}
u=T,\\[1mm]
\displaystyle
cT-\sigma_2v+\beta_2T\frac{x}{x+\eta_1}=0.
\end{array}
\right.
\end{equation}
Therefore,
\begin{equation}
v(x)
=
\frac{1}{\sigma_2}
\left(
cT+\beta_2T\frac{x}{x+\eta_1}
\right).
\label{eq:v_saturated_solution}
\end{equation}
This solution is retained only when it satisfies \eqref{eq:gamma2_above_one}.

When $x=0$, we have $a(0)=0$, so the quadratic formula cannot be used. Therefore, we consider $x=0$ separately. From \eqref{eq:qssa_equations},
\begin{equation}
u=T,
\qquad
v(0)=\frac{cT}{\sigma_2}.
\end{equation}

If $\eta_1=\eta_2=\eta$, the two functional responses coincide and
\eqref{eq:qssa_equations} reduces directly to the single-valued relation
\begin{equation}
v(x)=\frac{1}{\sigma_2}
\left(cT+\beta_2T\frac{x}{x+\eta}\right).
\label{eq:v-equal-eta}
\end{equation}
Thus, the quadratic formula \eqref{eq:v_quadratic_solutions} is needed
only when $\eta_1\neq\eta_2$; the case $\eta_1=\eta_2$, including
$\rho_\mu=0$, is covered by \eqref{eq:v-equal-eta}.

To make the reduced system single-valued, define
\[
\Phi(x,v)=
\frac{x}{x+\eta_1}\gamma_2\!\left(\frac{v}{\widehat{\beta}_2}\right)
+
\frac{x}{x+\eta_2}
\left[1-\gamma_2\!\left(\frac{v}{\widehat{\beta}_2}\right)\right]
\]
and
\[
\mathcal F(x,v)=cT-\sigma_2v+\beta_2T\Phi(x,v).
\]
We select a continuous nonnegative quasi-steady branch $v=v(x)$,
defined on an interval $\mathcal I\subset[0,\infty)$ containing the
biologically relevant range of $x$, such that
\[
\mathcal F(x,v(x))=0,
\qquad
v(0)=\frac{cT}{\sigma_2},
\qquad x\in\mathcal I.
\]
The selected value must satisfy the appropriate admissibility condition
\eqref{eq:gamma2_below_one} or \eqref{eq:gamma2_above_one}. If multiple
admissible roots exist, the same branch is used throughout. All results
below refer to this selected branch.

Substituting $u=T$ and the selected branch $v=v(x)$ into
the first two equations of \eqref{eq:nondimension}, we obtain the
following quasi-steady system:
\begin{equation}
\label{eq: final2D}
\resizebox{\textwidth}{!}{ 
    $\displaystyle
\begin{aligned}
&\begin{cases}
\frac{\mathrm{d}x}{\mathrm{d}t}=
\Bigl(\Bigl\{\alpha_{11}\gamma_1(\frac{T}{\widehat{\alpha}_{2}})+\alpha_{12} \bigl[1-\gamma_1(\frac{T}{\widehat{\alpha}_{2}})\bigl]\Bigl\}\frac{\max\{0,1-T\}}{1+x}-k_1T-m_1\Bigl)x-\Bigl\{\frac{x}{x+\eta_1}\gamma_2(\frac{v(x)}{\widehat{\beta}_{2}})+\frac{x}{x+\eta_2}\bigl[1-\gamma_2(\frac{v(x)}{\widehat{\beta}_{2}})\bigl]\Bigl\}y\\
\frac{dy}{dt}=\beta_1\max\{0,1-v(x)\}\Bigl\{\frac{x}{x+\eta_1}\gamma_2(\frac{v(x)}{\widehat{\beta}_{2}})+\frac{x}{x+\eta_2}\bigl[1-\gamma_2(\frac{v(x)}{\widehat{\beta}_{2}})\bigl]\Bigl\}y-(k_2v(x)+m_2)y \\
\end{cases}
\end{aligned}
$}
\end{equation}
\subsection{Positivity and boundedness}

For system \eqref{eq: final2D} to be biologically meaningful, the prey and predator populations must remain nonnegative and bounded. The following theorem establishes this result, and its proof is given in Appendix~\ref{app: Positivity and boundedness}.

\begin{theorem}[Positivity and Boundedness]
\label{th:positivity_boundedness}
Let
\begin{equation}
k=
\max\left\{
\frac{\alpha_{11}\beta_1}{\min\{m_1,m_2\}},
\,
\beta_1x(0)+y(0)
\right\}+1
\end{equation}
and define
\begin{equation}
\Omega=
\left\{
(x,y):
0\leq x\leq\frac{k}{\beta_1},
\quad
0\leq y\leq k
\right\}.
\end{equation}
Assume that the selected quasi-steady branch $v=v(x)$ is continuous for
$0\leq x\leq k/\beta_1$. Solutions of system~\eqref{eq: final2D} with
initial conditions in $\Omega$ remain in $\Omega$ for all $t\geq0$.
\end{theorem}

\subsection{Equilibrium existence and local stability}

Throughout the analysis, we assume that
\begin{equation}
m_1<1,\qquad
0<T<\min\left\{1,\frac{\sigma_2}{c}\right\},
\qquad
m_2<\beta_1.
\label{eq:basic_assumptions}
\end{equation}
The biological motivation for these assumptions is discussed in
\citep{huang2015impact}. We also assume that
\begin{equation}
\alpha_{12}\leq\alpha_{11},
\qquad
\mu_2\leq\mu_1,
\qquad
\eta_1\leq\eta_2.
\label{eq:personality_assumptions}
\end{equation}
The condition $\alpha_{12}\leq\alpha_{11}$ reflects our assumption that bold prey have a higher birth and growth rate \citep{biro2014individual,ballew2017fitness}. The conditions $\mu_2\leq\mu_1$ and $\eta_1\leq\eta_2$ reflect our assumption that bold predators have a higher encounter rate with prey.

For the nonnegative branch $v=v(x)$ selected above, define
\begin{equation}
\begin{aligned}
f(x)
&=
\left[
\left\{
\alpha_{11}\gamma_1\left(\frac{T}{\widehat{\alpha}_2}\right)
+\alpha_{12}
\left[
1-\gamma_1\left(\frac{T}{\widehat{\alpha}_2}\right)
\right]
\right\}
\frac{1-T}{1+x}
-k_1T-m_1
\right]x,\\
\phi(x)
&=
\frac{x}{x+\eta_1}
\gamma_2\left(\frac{v(x)}{\widehat{\beta}_2}\right)
+
\frac{x}{x+\eta_2}
\left[
1-\gamma_2\left(\frac{v(x)}{\widehat{\beta}_2}\right)
\right].
\end{aligned}
\label{eq:f_phi}
\end{equation}
From the second quasi-steady equation,
\begin{equation}
v(x)
=
\frac{cT}{\sigma_2}
+
\frac{\beta_2T}{\sigma_2}\phi(x).
\label{eq:v_phi}
\end{equation}
Therefore, we define
\begin{equation}
\begin{aligned}
g(x)
={}&
\beta_1\phi(x)
\max\left\{
0,\,
1-\frac{cT}{\sigma_2}
-\frac{\beta_2T}{\sigma_2}\phi(x)
\right\}-
\frac{k_2cT}{\sigma_2}
-
\frac{k_2\beta_2T}{\sigma_2}\phi(x)
-m_2.
\end{aligned}
\label{eq:g}
\end{equation}
System \eqref{eq: final2D} can then be written as
\begin{equation}
\left\{
\begin{array}{l}
\displaystyle
\frac{\mathrm{d}x}{\mathrm{d}t}
=f(x)-\phi(x)y,\\[2mm]
\displaystyle
\frac{\mathrm{d}y}{\mathrm{d}t}
=g(x)y.
\end{array}
\right.
\label{eq:reduced_system}
\end{equation}

\subsubsection{Conditions for the existence of equilibria}

In this section, we study how the toxin level $T$, the prey birth and
growth contrast $\rho_\alpha$, and the predator personality contrast
$\rho_\mu$ determine the existence of extinction, prey-only, and
coexistence equilibria of system~\eqref{eq: final2D}. The conditions
are evaluated in three steps.

First, following Appendix C of \citet{huang2015impact}, the toxin
thresholds $T_1^*$, $T_2^*$, and $T_3^*$ determine the number of
candidate coexistence equilibria. Their definitions are collected
after Table~\ref{tab:ex-of-eq}.

Second, for a fixed toxin level $T$, the prey personality contrast
$\rho_\alpha$ determines whether the prey can persist. When $\gamma_1\neq1/2$, 
\[
\begin{cases}
\rho_\alpha>\rho_\alpha^*(T),
& \gamma_1>\dfrac{1}{2},\\[2mm]
\rho_\alpha<\rho_\alpha^*(T),
& \gamma_1<\dfrac{1}{2}.
\end{cases}
\]
When $\gamma_1=1/2$, the weighted prey birth and growth coefficient is
independent of $\rho_\alpha$. In this case, prey persistence is equivalent
to
\[
\bar{\alpha}(1-T)>k_1T+m_1.
\]

Finally, the predator personality contrast $\rho_\mu$
determines the number of positive coexistence equilibria. For each
candidate root $\phi_i$, $i\in\{1,2\}$, the corresponding predator
body burden and bold proportion are
\begin{equation}
v_i=\frac{cT+\beta_2T\phi_i}{\sigma_2},
\qquad
\gamma_{2,i}
=
\gamma_2\left(\frac{v_i}{\widehat{\beta}_2}\right).
\label{eq:candidate-v-gamma}
\end{equation}
Under the prey-persistence condition above, an admissible candidate
$\phi_i$ gives a positive coexistence equilibrium if and only if
\begin{equation}
\Delta_{\mu,i}>0,
\qquad
\rho_\mu\in
[0,1)\cap
\left(\rho_{\mu,i}^{-},\rho_{\mu,i}^{+}\right),
\qquad
v(x_i)=v_i.
\label{eq:candidate-existence-conditions}
\end{equation}
Here, admissible means that $\phi_i$ is a distinct real root satisfying
the bounds in \eqref{eq:phi-candidate-condition}.
The interval must be nonempty. Depending on which candidates satisfy
these conditions, there are three possibilities: no positive coexistence
equilibrium, only $E_1$, or both $E_1$ and $E_2$.

The complete conditions are summarized in Table~\ref{tab:ex-of-eq}, and the detailed derivations are given in Appendix~\ref{app: Existence of equilibrium}.

\begin{table}[p]
\centering
\caption{Summary of the conditions for equilibrium existence and their biological interpretations. The notation is defined in Appendix~\ref{sec: Notation}.}
\label{tab:ex-of-eq}

\footnotesize
\renewcommand{\arraystretch}{1.10}
\setlength{\tabcolsep}{3.5pt}
\setlength{\abovedisplayskip}{3pt}
\setlength{\belowdisplayskip}{3pt}

\begin{tabularx}{\textwidth}{
    >{\centering\arraybackslash}p{0.04\textwidth}
    >{\raggedright\arraybackslash}p{0.55\textwidth}
    >{\raggedright\arraybackslash}X
}
\toprule
\textbf{Case}
&
\textbf{Mathematical conditions and equilibria}
&
\textbf{Biological interpretation}
\\
\midrule

\textbf{1}
&
If
\[
\begin{cases}
\rho_\alpha\leq\rho_\alpha^*(T),
& \gamma_1>\dfrac12,\\
\rho_\alpha\geq\rho_\alpha^*(T),
& \gamma_1<\dfrac12,
\end{cases}
\]
the system admits only
\[
E_{0,1}=(0,0).
\]
For $\gamma_1=1/2$, the same conclusion holds when
$\bar{\alpha}(1-T)\leq k_1T+m_1$.
&
The system has only the extinction equilibrium. When the majority
of prey individuals are bold $(\gamma_1>1/2)$, extinction occurs
when the prey personality contrast is sufficiently small. When
the majority are shy $(\gamma_1<1/2)$, extinction occurs when the
contrast is sufficiently large.
\\
\midrule

\textbf{2}
&
If
\[
\begin{cases}
\rho_\alpha>\rho_\alpha^*(T),
& \gamma_1>\dfrac12,\\
\rho_\alpha<\rho_\alpha^*(T),
& \gamma_1<\dfrac12,
\end{cases}
\]
the system always admits
\[
E_{0,1}=(0,0),
\qquad
E_{0,2}=(x_0,0).
\]
For $\gamma_1=1/2$, the same conclusion holds when
$\bar{\alpha}(1-T)>k_1T+m_1$. Cases 3 and 4 below specify when
coexistence equilibria are present in addition.
&
The system has extinction and prey-only equilibria. When the majority
of prey individuals are bold, prey persistence requires a sufficiently
large prey personality contrast. When the majority are shy, it
requires a sufficiently small contrast. Under these conditions,
coexistence is also possible if the additional
predator-nullcline conditions are satisfied.
\\
\midrule

\textbf{3}
&
Under the condition of Case 2, the system has one positive
coexistence equilibrium if either

\textnormal{(a)}
\[
T<\min\{T_1^*,T_2^*\},
\]
and
\[
\Delta_{\mu,1}>0,\qquad
\rho_\mu\in
[0,1)\cap
\left(\rho_{\mu,1}^{-},\rho_{\mu,1}^{+}\right),
\qquad
v(x_1)=v_1;
\]

or \textnormal{(b)}
\[
\max\{T_2^*,T_3^*\}<T<T_1^*,
\]
the above conditions hold for $i=1$, but the corresponding
conditions fail for $i=2$.

Hence, the equilibria are
\[
E_{0,1},\qquad
E_{0,2},\qquad
E_1=
\left(
x_1,\frac{f(x_1)}{\phi_1}
\right).
\]
&
If the condition in Case 2 is satisfied, one positive coexistence
equilibrium exists either at a sufficiently low toxin level when the
conditions for $E_1$ are satisfied, or at intermediate toxin levels
when the conditions for $E_1$ are satisfied but those for $E_2$ are
not.
\\
\midrule

\textbf{4}
&
Under the condition of Case 2, if
\[
\max\{T_2^*,T_3^*\}<T<T_1^*
\]
and
\[
\Delta_{\mu,2}>0,\qquad
\rho_\mu\in
[0,1)\cap
\left(\rho_{\mu,2}^{-},\rho_{\mu,2}^{+}\right),
\qquad
v(x_2)=v_2,
\]
then $E_2$ exists. By \eqref{eq:E2-implies-E1}, $E_1$ also exists.
Hence, the equilibria are
\[
E_{0,1},\qquad
E_{0,2},\qquad
E_1=
\left(
x_1,\frac{f(x_1)}{\phi_1}
\right),
\qquad
E_2=
\left(
x_2,\frac{f(x_2)}{\phi_2}
\right).
\]
&
If the condition in Case 2 is satisfied, two positive coexistence
equilibria exist at intermediate toxin levels when the predator
personality contrast lies within the range required for $E_2$. The
existence of $E_2$ also guarantees the existence of $E_1$.
\\
\bottomrule
\end{tabularx}
\end{table}

\subsubsection{Effect of prey personality contrast}

The biological meaning of $\rho_\alpha^*(T)$ follows from the
weighted prey birth and growth coefficient.
\begin{equation}
\bar{\alpha}
\left[
1+\rho_\alpha(2\gamma_1-1)
\right].
\label{eq:effective_prey_growth}
\end{equation}
Its derivative with respect to $\rho_\alpha$ is
\[
\bar{\alpha}(2\gamma_1-1).
\]

When $\gamma_1>1/2$, bold prey are the majority. Increasing
$\rho_\alpha$ increases the weighted prey birth and growth coefficient.
As $\rho_\alpha$ crosses $\rho_\alpha^*(T)$ from below, the system
changes from Case 1 to one of Cases 2-4. Thus, the system changes from
extinction only to a state with a positive prey-only equilibrium and
possible coexistence equilibria.

When $\gamma_1<1/2$, bold prey are the minority. Increasing
$\rho_\alpha$ decreases the weighted prey birth and growth coefficient.
Crossing $\rho_\alpha^*(T)$ from below changes the system from one of
Cases 2-4 to Case 1. Therefore, the effect of $\rho_\alpha$ depends on
whether bold prey are the majority or minority of the prey population. When $\gamma_1=1/2$, the sign of
$x_0$ is independent of $\rho_\alpha$ and must be checked directly
from $x_0=\bar{\alpha}(1-T)/(k_1T+m_1)-1$.

Figure~\ref{fig:prey_equilibrium_regions} illustrates the
contrasting effects of $\rho_\alpha$ described above together with the
corresponding local-stability regimes. At each parameter point, all
biologically feasible equilibria were first computed, and their local
stability was determined from the eigenvalues of the corresponding
Jacobian matrices. The red region represents a locally stable
extinction equilibrium $E_{0,1}$, the blue region a locally stable prey-only
equilibrium $E_{0,2}$, and the green region a locally stable coexistence
equilibrium $E_1$. The yellow region represents local
bistability between the prey-only equilibrium and $E_1$. The
coexistence equilibrium $E_2$ was also examined, but no locally stable
$E_2$ was found in the displayed parameter ranges. The dark curves
separate the numerically determined local-stability regimes, whereas
the dashed line represents the existence threshold
$\rho_\alpha=\rho_\alpha^*(T)$.

The two panels illustrate how the effect of prey personality
contrast depends on the relative abundance of bold and shy prey. In
panel (a), $\gamma_1=0.407<1/2$, so shy prey are the majority.
Increasing $\rho_\alpha$ across $\rho_\alpha^*(T)$ moves the system
from a parameter region with a locally stable prey-only or coexistence
equilibrium into the locally stable extinction region. Biologically,
a larger $\rho_\alpha$ represents a stronger personality-induced
contrast in the birth and growth rates of bold and shy prey. Because
the majority of prey are shy and have the lower birth and growth rate,
a larger contrast reduces the weighted prey birth and growth
coefficient and weakens prey persistence. In contrast, in panel (b), $\gamma_1=0.525>1/2$, so bold prey are the
majority. Increasing $\rho_\alpha$ across $\rho_\alpha^*(T)$ moves the
system from the locally stable extinction region into a parameter
region with a locally stable prey-only or coexistence equilibrium.
The opposite biological effect occurs because the majority of prey
are bold and have the higher birth and growth rate. Therefore, a
larger personality-induced contrast increases the weighted prey birth
and growth coefficient and promotes prey persistence. Taken together, increasing prey personality contrast promotes extinction when shy
prey are the majority but promotes persistence when bold prey are the
majority. Within the prey-persistence range, the joint effects of
$\rho_\alpha$ and $\rho_\mu$ further determine whether the system
exhibits locally stable prey-only persistence, locally stable
coexistence, or local bistability.

\begin{figure}[htbp]
\centering
\includegraphics[width=0.96\textwidth]{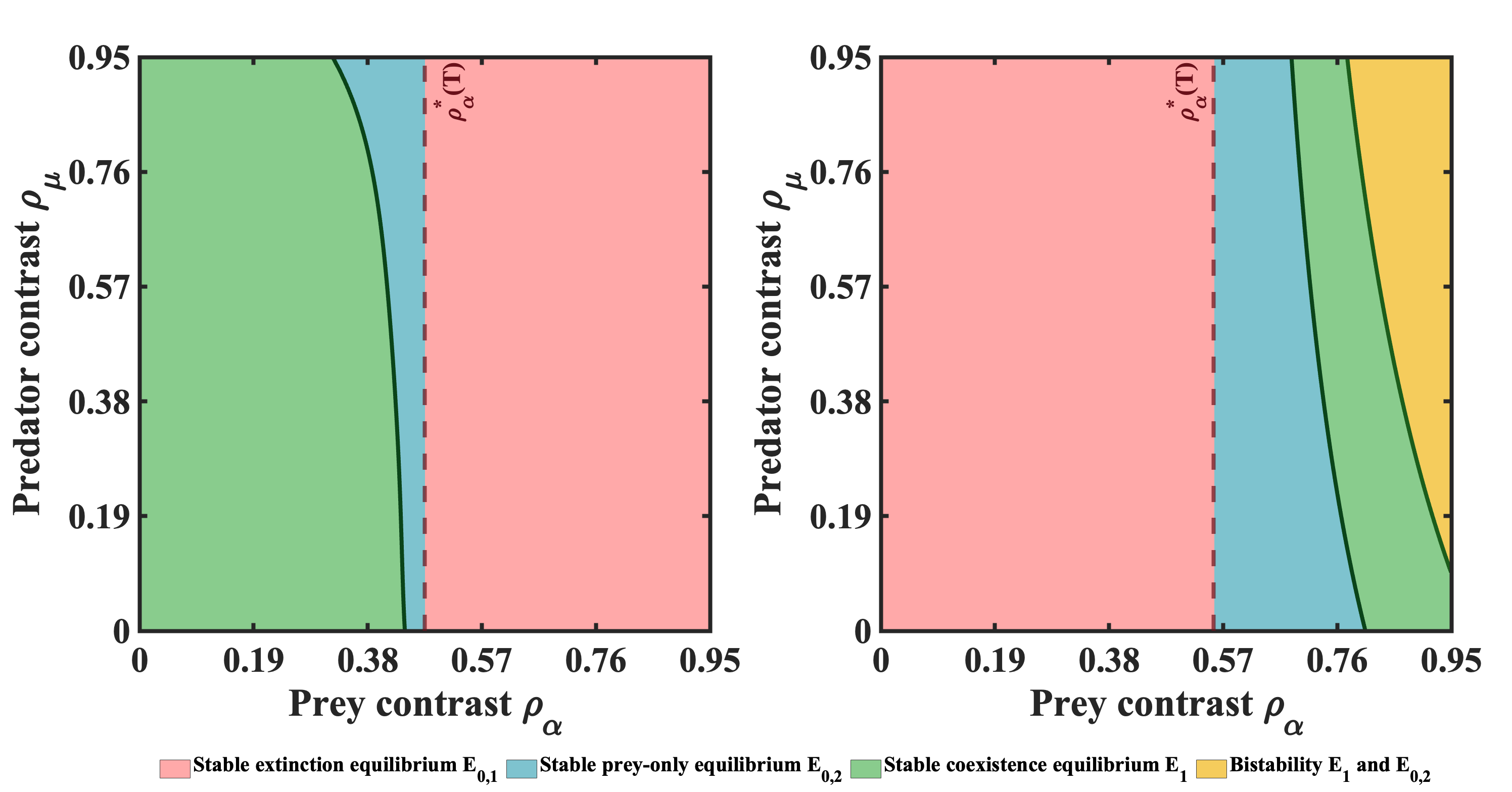}
\caption{Existence and local-stability classification in the
$(\rho_\alpha,\rho_\mu)$ plane for two representative cases:
(a) $\gamma_1<1/2$, and (b) $\gamma_1>1/2$. The parameters in panel (a)
are $\widehat{\beta}_2=33.4100$, $T=0.11$, $\widehat{\alpha}_2=0.4250$, $\bar{\alpha}=0.2000$,
$m_1=0.1619$, $\beta_1=1.2648$, $m_2=0.0075$,
$k_1=0.0039$, $k_2=0.3079$, $\sigma_2=2.5805$,
$\beta_2=23.2793$, $c=0.9328$, and $\bar{\eta}=0.3017$.
The parameters in panel (b) are $\widehat{\beta}_2=33.4100$, $T=0.20$, $\widehat{\alpha}_2=0.4000$,
$\bar{\alpha}=0.2380$, $m_1=0.1912$, $\beta_1=4.8421$,
$m_2=0.0104$, $k_1=0.0229$, $k_2=0.7007$,
$\sigma_2=3.1101$, $\beta_2=52.0996$, $c=1.2091$, and
$\bar{\eta}=0.2387$.}
\label{fig:prey_equilibrium_regions}
\end{figure}

\subsubsection{Effect of predator personality contrast}

Under the prey-persistence condition, we fix
$\rho_\alpha=0.20$ and study how the predator personality contrast
$\rho_\mu$ affects the number of positive coexistence equilibria.

For each candidate root $\phi_i$, a positive coexistence equilibrium
exists if
\[
\Delta_{\mu,i}>0,
\qquad
\max\{0,\rho_{\mu,i}^{-}\}
<
\rho_\mu
<
\min\{1,\rho_{\mu,i}^{+}\},
\qquad
v(x_i)=v_i.
\]
The interval condition shows that the existence of a candidate can
depend on $\rho_\mu$. In general, if both endpoints lie in $(0,1)$,
the candidate exists only for an intermediate range of predator
personality contrast.

The intermediate range of $\rho_\mu$ follows from the positivity
condition $x_i<x_0$. As $\rho_\mu$ increases, the contribution of bold
predators to the functional response increases, whereas that of shy
predators decreases. The balance between these two effects satisfies
$x_i<x_0$ only when
\[
\rho_{\mu,i}^-<\rho_\mu<\rho_{\mu,i}^+.
\]
Therefore, when both bounds lie in $(0,1)$, positive coexistence occurs
only at intermediate values of the predator personality contrast.

The conditions above show that, once prey persistence is possible,
$\rho_\mu$ primarily controls the existence and number of positive
coexistence equilibria. Figure~\ref{fig:coexistence_equilibrium_regions}
shows the three possibilities listed in Table~\ref{tab:ex-of-eq}
together with their local stability. At each point of the
$(T,\rho_\mu)$ plane, the analytical conditions were first used to
determine whether there is no positive coexistence equilibrium, only
$E_1$, or both $E_1$ and $E_2$. The local stability of all feasible
equilibria was then determined numerically from the eigenvalues of the
corresponding Jacobian matrices.

The red and blue regions contain no positive coexistence equilibrium;
they represent a locally stable extinction equilibrium $E_{0,1}$ and
a locally stable prey-only equilibrium $E_{0,2}$, respectively. In
the green region, only the positive coexistence equilibrium $E_1$
exists and it is locally stable. In the yellow region, both
$E_1$ and $E_2$ exist. Here, $E_1$ and the prey-only equilibrium
$E_{0,2}$ are locally stable, giving local bistability, whereas $E_2$
is unstable. The boundaries between these regions agree with
the existence conditions derived above, while the
numerical calculation further determines the stability of the
feasible equilibria.

The curved boundary between the green and yellow regions shows
the effect of $\rho_\mu$ clearly. At intermediate toxin levels,
increasing $\rho_\mu$ changes the number of positive coexistence
equilibria from one to two. At the same time, the local-stability
regime changes from stable coexistence alone to bistability between
$E_1$ and $E_{0,2}$. For a horizontal slice passing through the
yellow region, increasing the toxin level changes the system from
one positive coexistence equilibrium to two, then back to one, and
finally to none. The equilibrium stability pattern changes from local monostability to local bistability and then back to local monostability. The vertical lines
indicate the analytical toxin thresholds $T_3^*$, $T_2^*$, and
$T_1^*$.

\begin{figure}[htbp]
\centering
\includegraphics[width=0.96\textwidth]{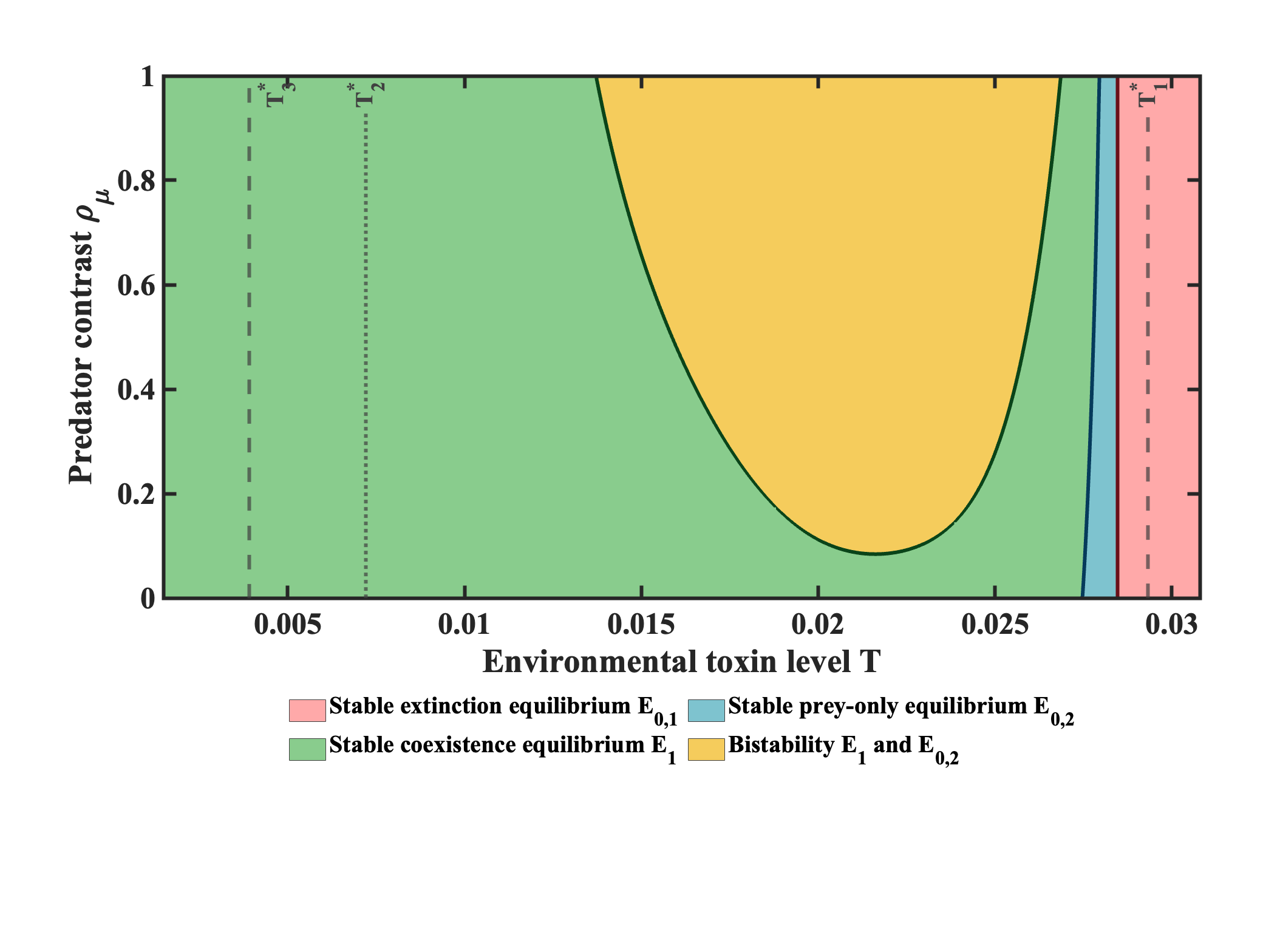}
\caption{Existence and local-stability classification in the
$(T,\rho_\mu)$ plane for fixed $\rho_\alpha=0.20$. All
parameter values are reported to four decimal places:
$m_1=0.3617$, $\beta_1=0.8688$, $m_2=0.0079$,
$k_1=1.9610$, $k_2=0.0090$, $\sigma_2=0.1545$,
$\beta_2=18.2991$, $c=2.7205$, $\bar{\eta}=0.9435$,
$\bar{\alpha}=0.5328$, $\widehat{\alpha}_2=0.0220$, and
$\widehat{\beta}_2=33.4100$.}
\label{fig:coexistence_equilibrium_regions}
\end{figure}

\FloatBarrier
\section{Discussion}
\label{sec: Discussion}

Many experimental studies have shown that animal personality can influence behaviour and may also be affected by environmental toxins \citep{Cole2014,Michelena2009sheep,Scott2004behavior,RamirezRomero2005}. Motivated by these experimental observations, we developed a predator–prey model that incorporates toxin-induced intraspecific personality variation in both prey and predators. 


Our results show that prey and predator personality have different effects on long-term population dynamics. When prey personality is included, the oscillatory intervals are reduced or disappear. In contrast, predator personality can generate new oscillatory intervals or split one interval into two. We also find that increasing the difference in prey birth and growth reduces population oscillations, whereas increasing the difference in predator foraging amplifies them. When personality is included in both prey and predators, their effects are not simply additive. This may be because prey personality affects birth and growth, whereas predator personality affects foraging and predation pressure.

Using a quasi-steady-state approximation, we reduced the original
model to a two-dimensional system. We established its positivity and
boundedness and derived conditions for the existence of extinction,
prey-only, and positive coexistence equilibria. The equilibrium
structure is jointly determined by the toxin level $T$, the prey
personality contrast $\rho_\alpha$, and the predator personality
contrast $\rho_\mu$. The toxin thresholds $T_1^*$, $T_2^*$, and
$T_3^*$ determine the possible candidate roots for coexistence. For a
fixed $T$, $\rho_\alpha$ determines whether the prey can persist.
Increasing $\rho_\alpha$ promotes prey persistence when bold prey are
the majority but has the opposite effect when shy prey are the
majority. The predator personality contrast $\rho_\mu$ determines
whether the system has zero, one, or two positive coexistence
equilibria. Based on this existence classification, we numerically
investigated the local stability of all biologically feasible
equilibria and identified parameter regions corresponding to locally
stable extinction, prey-only persistence, coexistence, and
bistability.


The effects of personality are mainly observed at low and moderate toxin concentrations, while sufficiently high toxin concentrations lead to predator extinction in all model variants. This suggests that personality can change population dynamics before direct toxic effects cause population loss. Our model assumes a constant toxin concentration, although toxin levels in natural aquatic ecosystems may vary over time and space \citep{huang2015impact}. Moreover, individual personality may influence biological metabolism \citep{Careau2008dis}. Future work could therefore consider time- and space-dependent toxin exposure, alternative personality response functions, and personality-dependent metabolism. Further analytical and global stability analyses, together with bifurcation analysis, are also needed. Overall, our results show that toxin-induced personality variation can qualitatively change predator-prey dynamics.

\appendix
\section{Proofs and supplementary fitting details}

\subsection{Positivity and boundedness}
\label{app: Positivity and boundedness}

\begin{proof}

Assume that $S(t)=(x(t),y(t))$ is a solution of system
\eqref{eq: final2D} with $S(0)\in\Omega$. Let $t_1$ be the first
time that $S(t)$ reaches or crosses the boundary of $\Omega$.
We consider the following three cases.

\medskip
\noindent
\textbf{Case 1.} Suppose that $x(t_1)=0$. For
$t\in[0,t_1]$, we have
\begin{equation}
\begin{aligned}
\frac{\mathrm{d}x}{\mathrm{d}t}
={}&
\left[
\alpha_{11}\gamma_1
\left(\frac{T}{\widehat{\alpha}_2}\right)
+
\alpha_{12}
\left(
1-\gamma_1
\left(\frac{T}{\widehat{\alpha}_2}\right)
\right)
\right]
\frac{\max\{0,1-T\}}{1+x}x-(k_1T+m_1)x\\
&-
\left[
\frac{x}{x+\eta_1}
\gamma_2\left(\frac{v}{\widehat{\beta}_2}\right)
+
\frac{x}{x+\eta_2}
\left(
1-\gamma_2\left(\frac{v}{\widehat{\beta}_2}\right)
\right)
\right]y.
\end{aligned}
\end{equation}
Since the first term is nonnegative,
$\eta_1\leq\eta_2$, and $y\leq k$ on $[0,t_1]$,
\begin{equation}
\begin{aligned}
\frac{\mathrm{d}x}{\mathrm{d}t}
&\geq
-(k_1T+m_1)x-\frac{x}{\eta_1}y\geq
\left(
-k_1T-m_1-\frac{k}{\eta_1}
\right)x\equiv\delta_1x,
\end{aligned}
\end{equation}
where $\delta_1=-k_1T-m_1-\frac{k}{\eta_1}$. Therefore, $x(t_1)\geq x(0)e^{\delta_1t_1}>0$, which contradicts $x(t_1)=0$. Hence, the solution cannot cross
the boundary $x=0$.

\medskip
\noindent
\textbf{Case 2.} Suppose that $y(t_1)=0$. Since $v(t)$ is continuous
on $[0,t_1]$, let $M_v:=\max_{s\in[0,t_1]} v(s).$ Then, for $t\in[0,t_1]$,
\begin{equation}
\begin{aligned}
\frac{\mathrm{d}y}{\mathrm{d}t}
={}&
\beta_1\max\{0,1-v\}
\left[
\frac{x}{x+\eta_1}
\gamma_2\left(\frac{v}{\widehat{\beta}_2}\right)
+
\frac{x}{x+\eta_2}
\left(
1-\gamma_2\left(\frac{v}{\widehat{\beta}_2}\right)
\right)
\right]y-(k_2v+m_2)y\\
&\geq-(k_2M_v+m_2)y\\
&\equiv\delta_2y,
\end{aligned}
\end{equation}
where $\delta_2=-k_2M_v-m_2$.
Therefore, $y(t_1)\geq y(0)e^{\delta_2t_1}>0$,
which contradicts $y(t_1)=0$. Hence, the solution cannot cross
the boundary $y=0$.

\medskip
\noindent
\textbf{Case 3.} Suppose that $x(t_1)=\frac{k}{\beta_1}$ or $y(t_1)=k$. Let $z(t)=\beta_1x(t)+y(t)$. Then
\begin{equation}
\begin{aligned}
\frac{\mathrm{d}z}{\mathrm{d}t}
={}&
\beta_1
\left[
\alpha_{11}\gamma_1
\left(\frac{T}{\widehat{\alpha}_2}\right)
+
\alpha_{12}
\left(
1-\gamma_1
\left(\frac{T}{\widehat{\alpha}_2}\right)
\right)
\right]
\frac{\max\{0,1-T\}}{1+x}x\\
&-\beta_1(k_1T+m_1)x-(k_2v+m_2)y\\
&+\beta_1
\left[\max\{0,1-v\}-1\right]
\left[
\frac{x}{x+\eta_1}
\gamma_2\left(\frac{v}{\widehat{\beta}_2}\right)
+
\frac{x}{x+\eta_2}
\left(
1-\gamma_2\left(\frac{v}{\widehat{\beta}_2}\right)
\right)
\right]y.
\end{aligned}
\end{equation}
Since
\[
\max\{0,1-T\}\leq1,\qquad
\max\{0,1-v\}\leq1,\qquad
\alpha_{12}\leq\alpha_{11},
\]
we obtain
\begin{equation}
\begin{aligned}
\frac{\mathrm{d}z}{\mathrm{d}t}
&\leq
\beta_1\frac{\alpha_{11}}{1+x}x
-\beta_1m_1x-m_2y\\
&\leq
\alpha_{11}\beta_1
-\beta_1m_1x-m_2y\\
&\leq
\alpha_{11}\beta_1
-\min\{m_1,m_2\}z.
\end{aligned}
\end{equation}
Solving this differential inequality gives
\begin{equation}
z(t)
\leq
\left[
z(0)-
\frac{\alpha_{11}\beta_1}{\min\{m_1,m_2\}}
\right]
e^{-\min\{m_1,m_2\}t}
+
\frac{\alpha_{11}\beta_1}{\min\{m_1,m_2\}}.
\end{equation}
Hence,
\[
z(t)
\leq
\max\left\{
\frac{\alpha_{11}\beta_1}{\min\{m_1,m_2\}},
\beta_1x(0)+y(0)
\right\}
<k.
\]
However, either $x(t_1)=k/\beta_1$ or $y(t_1)=k$ implies
$z(t_1)\geq k$, which is a contradiction. Therefore, the solution
cannot cross either upper boundary.

Thus, every solution starting in $\Omega$ remains in $\Omega$
for all forward time. This completes the proof.

\end{proof}

\subsection{Existence of equilibria}
\label{app: Existence of equilibrium}

\subsubsection*{Case 1: Extinction equilibrium only}

Since the $x$-nullcline contains $x=0$ and the $y$-nullcline
contains $y=0$, their intersection gives the extinction equilibrium
\[
E_{0,1}=(0,0).
\]
Since $g(0)=-k_2v(0)-m_2<0$, no equilibrium with $x=0$ and
$y>0$ exists. Hence any additional nonnegative equilibrium must have
$x>0$ and lie on the nullcline
\[
y=\frac{f(x)}{\phi(x)}
\]
together with either $g(x)=0$ or $y=0$.

Since $\phi(x)$ is positive on $(0,\infty)$, the sign of
$f(x)/\phi(x)$ is determined by the sign of $f(x)$. We therefore
consider
\begin{equation}
\begin{aligned}
f(x)=
\Bigg(
&\left[
\alpha_{11}\gamma_1
\left(\frac{T}{\widehat{\alpha}_2}\right)
+
\alpha_{12}
\left(
1-\gamma_1
\left(\frac{T}{\widehat{\alpha}_2}\right)
\right)
\right]
\frac{1-T}{1+x} -k_1T-m_1
\Bigg)x.
\end{aligned}
\end{equation}
Since $x>0$, it remains to study the sign of
\begin{equation}
\begin{aligned}
p(x):={}&
\left[
\alpha_{11}\gamma_1
\left(\frac{T}{\widehat{\alpha}_2}\right)
+
\alpha_{12}
\left(
1-\gamma_1
\left(\frac{T}{\widehat{\alpha}_2}\right)
\right)
\right]
\frac{1-T}{1+x}-k_1T-m_1.
\end{aligned}
\end{equation}
The function $p(x)$ is decreasing with respect to $x$. Solving
$p(x)=0$ gives
\begin{equation}
x_0=
\left[
\alpha_{11}\gamma_1
\left(\frac{T}{\widehat{\alpha}_2}\right)
+
\alpha_{12}
\left(
1-\gamma_1
\left(\frac{T}{\widehat{\alpha}_2}\right)
\right)
\right]
\frac{1-T}{k_1T+m_1}-1.
\end{equation}
Thus, if $x_0\leq0$, then
\begin{equation}
p(x)<p(0)\leq0,
\qquad \forall x>0.
\end{equation}
Consequently, $f(x)<0$ for all $x>0$, and hence the nullcline
\[
y=\frac{f(x)}{\phi(x)}
\]
does not intersect either $g(x)=0$ or $y=0$ in the positive region.

Next, we derive a threshold condition on $\rho_\alpha$. Since
\[
\begin{aligned}
&\alpha_{11}\gamma_1
\left(\frac{T}{\widehat{\alpha}_2}\right)
+
\alpha_{12}
\left[
1-\gamma_1
\left(\frac{T}{\widehat{\alpha}_2}\right)
\right]=
\bar{\alpha}(1+\rho_\alpha)\gamma_1
+
\bar{\alpha}(1-\rho_\alpha)(1-\gamma_1)=
\bar{\alpha}
\left(1-\rho_\alpha+2\rho_\alpha\gamma_1\right),
\end{aligned}
\]
we have
\begin{equation}
x_0=
\bar{\alpha}
\left(1-\rho_\alpha+2\rho_\alpha\gamma_1\right)
\frac{1-T}{k_1T+m_1}-1
\leq0.
\end{equation}
Therefore, for $\gamma_1\neq1/2$, the condition $x_0\leq0$ is
equivalent to
\begin{equation}
\begin{cases}
\displaystyle
\rho_\alpha
\leq
\frac{
\dfrac{k_1T+m_1}{\bar{\alpha}(1-T)}-1
}{
2\gamma_1-1
}
:=\rho_\alpha^*(T),
&
\displaystyle \gamma_1>\frac12,
\\[4mm]
\displaystyle
\rho_\alpha
\geq
\frac{
\dfrac{k_1T+m_1}{\bar{\alpha}(1-T)}-1
}{
2\gamma_1-1
}
:=\rho_\alpha^*(T),
&
\displaystyle \gamma_1<\frac12.
\end{cases}
\end{equation}

When $\gamma_1=1/2$, the weighted prey birth and growth coefficient
is $\alpha_{11}\gamma_1+\alpha_{12}(1-\gamma_1)
=
\frac{\alpha_{11}+\alpha_{12}}{2}
=
\bar{\alpha}$. Therefore, $x_0=
\bar{\alpha}\frac{1-T}{k_1T+m_1}-1$,
which is independent of $\rho_\alpha$. In this case, the system has
only the extinction equilibrium if $x_0\leq0$, while the prey-only
equilibrium also exists if $x_0>0$.

Hence, whenever $x_0\leq0$, the system admits only the
extinction equilibrium
\[
E_{0,1}=(0,0).
\]

\subsubsection*{Case 2: Extinction and prey-only equilibria}

According to the discussion in Case 1, the prey-only equilibrium
$E_{0,2}=(x_0,0)$ exists if $x_0>0$. Similarly, we derive a
threshold condition on $\rho_\alpha$. For $\gamma_1\neq1/2$, we have
\begin{equation}
\begin{cases}
\displaystyle
\rho_\alpha>\rho_\alpha^*(T),
&
\displaystyle \gamma_1>\frac12,
\\[2mm]
\displaystyle
\rho_\alpha<\rho_\alpha^*(T),
&
\displaystyle \gamma_1<\frac12.
\end{cases}
\end{equation}
When $\gamma_1=1/2$, the same conclusion holds precisely when
$\bar{\alpha}(1-T)>k_1T+m_1$. Hence, under this threshold condition, the system
admits the
extinction equilibrium
\[
E_{0,1}=(0,0)
\]
and the prey-only equilibrium
\[
E_{0,2}=(x_0,0).
\]

\subsubsection*{Case 3: Coexistence equilibria}

Throughout this proof, $v=v(x)$ denotes the nonnegative quasi-steady
branch selected above. In particular, the same
single-valued branch is used when evaluating every candidate
equilibrium below.

The coexistence equilibria satisfy
\begin{equation}
g(x)=0,
\qquad
y=\frac{f(x)}{\phi(x)}>0,
\qquad
x>0.
\end{equation}
Since $\phi(x)>0$ for $x>0$, we require
\begin{equation}
f(x)=
\left(
\left[\alpha_{11}\gamma_1+\alpha_{12}(1-\gamma_1)\right]
\frac{1-T}{1+x}
-k_1T-m_1
\right)x>0,
\end{equation}
where
\[
\gamma_1=\gamma_1\left(\frac{T}{\widehat{\alpha}_2}\right).
\]
This is equivalent to
\[
x_0>0,
\qquad
0<x<x_0,
\]
where
\begin{equation}
\begin{aligned}
x_0
&=
\left[\alpha_{11}\gamma_1+\alpha_{12}(1-\gamma_1)\right]
\frac{1-T}{k_1T+m_1}-1\\
&=
\bar{\alpha}
\left[1+\rho_\alpha(2\gamma_1-1)\right]
\frac{1-T}{k_1T+m_1}-1.
\end{aligned}
\end{equation}

For $\gamma_1\neq1/2$, the condition $x_0>0$ is equivalent to
\begin{equation}
\begin{cases}
\rho_\alpha>\rho_\alpha^*(T),
& \gamma_1>\dfrac12,\\[2mm]
\rho_\alpha<\rho_\alpha^*(T),
& \gamma_1<\dfrac12,
\end{cases}
\label{eq:coexistence-prey-condition}
\end{equation}
where
\begin{equation}
\rho_\alpha^*(T)=
\frac{
\dfrac{k_1T+m_1}{\bar{\alpha}(1-T)}-1
}{
2\gamma_1-1
}.
\end{equation}
When $\gamma_1=1/2$, the condition becomes
\[
\bar{\alpha}(1-T)>k_1T+m_1
\]
and is independent of the prey personality contrast
$\rho_\alpha$.

From the quasi-steady equation,
\begin{equation}
v(x)=\frac{cT+\beta_2T\phi(x)}{\sigma_2}.
\label{eq:v-phi-relation}
\end{equation}
If $v(x)\geq1$, the predator reproduction term is zero and
\[
g(x)=-(k_2v(x)+m_2)<0.
\]
Thus, a coexistence equilibrium must satisfy
\begin{equation}
\phi(x)<\frac{\sigma_2-cT}{\beta_2T}.
\label{eq:p1}
\end{equation}

Under \eqref{eq:p1}, we have
\begin{equation}
\begin{aligned}
g(x)={}&
\beta_1\phi(x)
\left(
1-\frac{cT}{\sigma_2}
-\frac{\beta_2T}{\sigma_2}\phi(x)
\right)-\frac{k_2cT}{\sigma_2}
-\frac{k_2\beta_2T}{\sigma_2}\phi(x)
-m_2.
\end{aligned}
\label{eq:g-reduced}
\end{equation}
Setting $g(x)=0$ gives
\begin{equation}
\beta_1\beta_2T\phi^2
-
\left[
\beta_1(\sigma_2-cT)-k_2\beta_2T
\right]\phi
+k_2cT+m_2\sigma_2=0.
\label{eq:phi-quadratic}
\end{equation}

The quadratic equation \eqref{eq:phi-quadratic} is the same as
that derived in Appendix C of \citet{huang2015impact}. Here we retain
the discriminant and admissibility conditions explicitly so that the
threshold cases are also covered.

The roots of \eqref{eq:phi-quadratic} are
\begin{equation}
\phi_{1,2}=
\frac{
\beta_1(\sigma_2-cT)-k_2\beta_2T
\mp\sqrt{\Delta(T)}
}{
2\beta_1\beta_2T
},
\label{eq:phi-candidate-roots}
\end{equation}
where
\begin{equation}
\Delta(T)=
\left[
\beta_1(\sigma_2-cT)-k_2\beta_2T
\right]^2
-4\beta_1\beta_2T
\left(
k_2cT+m_2\sigma_2
\right).
\end{equation}
Only distinct real roots satisfying
\begin{equation}
0<\phi_i<1,
\qquad
\phi_i<\frac{\sigma_2-cT}{\beta_2T}
\label{eq:phi-candidate-condition}
\end{equation}
are considered.

Following Appendix C of \citet{huang2015impact}, in the generic
nondegenerate parameter regions, only $\phi_1$ is admissible when
\[
T<\min\{T_1^*,T_2^*\},
\]
whereas both $\phi_1$ and $\phi_2$ are admissible when
\[
\max\{T_2^*,T_3^*\}<T<T_1^*.
\]

For each root $\phi_i$, the corresponding predator body burden
and bold proportion are
\begin{equation}
v_i=
\frac{cT+\beta_2T\phi_i}{\sigma_2},
\qquad
\gamma_{2,i}
=
\gamma_2\left(
\frac{v_i}{\widehat{\beta}_2}
\right).
\label{eq:equilibrium-v-gamma}
\end{equation}
The corresponding prey biomass $x_i$ satisfies
\begin{equation}
\phi_i=
\gamma_{2,i}\frac{x_i}{x_i+\eta_1}
+
(1-\gamma_{2,i})
\frac{x_i}{x_i+\eta_2}.
\label{eq:xi-functional-response}
\end{equation}
Multiplying by $(x_i+\eta_1)(x_i+\eta_2)$ gives
\begin{equation}
(1-\phi_i)x_i^2
+A_ix_i
-\phi_i\eta_1\eta_2=0,
\end{equation}
where
\begin{equation}
A_i=
\eta_1+(\eta_2-\eta_1)\gamma_{2,i}
-\phi_i(\eta_1+\eta_2).
\end{equation}
Since $0<\phi_i<1$, the leading coefficient is positive and
the constant term is negative. Thus, the two roots have opposite
signs, and the unique positive root is
\begin{equation}
x_i=
\frac{
-A_i+
\sqrt{
A_i^2+4\phi_i(1-\phi_i)\eta_1\eta_2
}
}{
2(1-\phi_i)
}.
\label{eq:xi-corrected}
\end{equation}

The candidate gives a positive coexistence equilibrium when
\[
0<x_i<x_0
\]
and
\begin{equation}
v(x_i)=v_i.
\label{eq:coexistence-branch-condition}
\end{equation}
The last condition ensures that $(x_i,v_i)$ lies on the
quasi-steady branch selected in the reduced system. It holds
automatically when the nonnegative quasi-steady solution is unique.
Under these conditions,
\begin{equation}
E_i=
\left(
x_i,\frac{f(x_i)}{\phi_i}
\right)
\end{equation}
is a positive coexistence equilibrium.

We next express $x_i<x_0$ in terms of the predator foraging
contrast $\rho_\mu$. Recall that
\begin{equation}
\mu_1=\bar{\mu}(1+\rho_\mu),
\qquad
\mu_2=\bar{\mu}(1-\rho_\mu),
\qquad
0\leq\rho_\mu<1.
\label{eq:mu-parameterization}
\end{equation}
Since $\eta_j=1/(\mu_jh)$, defining
\[
\bar{\eta}=\frac{1}{\bar{\mu}h}
\]
gives
\begin{equation}
\eta_1=\frac{\bar{\eta}}{1+\rho_\mu},
\qquad
\eta_2=\frac{\bar{\eta}}{1-\rho_\mu}.
\label{eq:eta-mu-parameterization}
\end{equation}

For each $\phi_i$, $\gamma_{2,i}$ is fixed by
\eqref{eq:equilibrium-v-gamma}. Therefore,
\begin{equation}
\begin{aligned}
\frac{\mathrm{d}}{\mathrm{d}x}
\left[
\gamma_{2,i}\frac{x}{x+\eta_1}
+
(1-\gamma_{2,i})\frac{x}{x+\eta_2}
\right]
={}&
\gamma_{2,i}\frac{\eta_1}{(x+\eta_1)^2}+
(1-\gamma_{2,i})
\frac{\eta_2}{(x+\eta_2)^2}
>0.
\end{aligned}
\end{equation}
Hence, \eqref{eq:xi-functional-response} gives
\begin{equation}
x_i<x_0
\quad\Longleftrightarrow\quad
\phi_i<
\gamma_{2,i}\frac{x_0}{x_0+\eta_1}
+
(1-\gamma_{2,i})\frac{x_0}{x_0+\eta_2}.
\label{eq:xi-x0-condition}
\end{equation}

Substituting \eqref{eq:eta-mu-parameterization}, the
right-hand side of \eqref{eq:xi-x0-condition} becomes
\begin{equation}
\frac{
x_0(\bar{\eta}+x_0)
+(2\gamma_{2,i}-1)\bar{\eta}x_0\rho_\mu
-x_0^2\rho_\mu^2
}{
(\bar{\eta}+x_0)^2-x_0^2\rho_\mu^2
}.
\label{eq:response-at-x0}
\end{equation}
For $x_0>0$ and $0\leq\rho_\mu<1$, the denominator is
positive because
\begin{equation}
(\bar{\eta}+x_0)^2-x_0^2\rho_\mu^2
>
\bar{\eta}^2+2\bar{\eta}x_0>0.
\end{equation}
Therefore,
\begin{equation}
x_i<x_0
\quad\Longleftrightarrow\quad
Q_i(\rho_\mu)>0,
\end{equation}
where
\begin{equation}
Q_i(\rho_\mu)=
-(1-\phi_i)x_0^2\rho_\mu^2
+(2\gamma_{2,i}-1)\bar{\eta}x_0\rho_\mu
+C_i
\label{eq:Qi-mu-positive}
\end{equation}
and
\begin{equation}
C_i=
(1-\phi_i)x_0^2
+(1-2\phi_i)\bar{\eta}x_0
-\phi_i\bar{\eta}^2.
\end{equation}
For fixed $T$, $\rho_\alpha$, $\bar{\mu}$, and the other
parameters, $\phi_i$, $\gamma_{2,i}$, and $x_0$ do not depend
on $\rho_\mu$. Hence, $Q_i$ is a quadratic in $\rho_\mu$.

The discriminant of $Q_i$ is
\begin{equation}
\Delta_{\mu,i}
=
x_0^2
\left[
\bar{\eta}^2(2\gamma_{2,i}-1)^2
+4(1-\phi_i)C_i
\right].
\end{equation}
Since $0<\phi_i<1$, $Q_i$ is a downward-opening quadratic.
If $\Delta_{\mu,i}\leq0$, the inequality
$Q_i(\rho_\mu)>0$ has no solution.

If $\Delta_{\mu,i}>0$, the two distinct roots are
\begin{equation}
\rho_{\mu,i}^{\pm}
=
\frac{
\bar{\eta}(2\gamma_{2,i}-1)
\pm
\sqrt{
\bar{\eta}^2(2\gamma_{2,i}-1)^2
+4(1-\phi_i)C_i
}
}{
2(1-\phi_i)x_0
}.
\label{eq:mu-distinct-roots}
\end{equation}
Consequently,
\begin{equation}
x_i<x_0
\quad\Longleftrightarrow\quad
\Delta_{\mu,i}>0
\quad\text{and}\quad
\rho_\mu\in
[0,1)\cap
\left(
\rho_{\mu,i}^{-},
\rho_{\mu,i}^{+}
\right).
\label{eq:mu-complete-condition}
\end{equation}
For $0<\rho_\mu<1$, the interval condition can also be
written as
\begin{equation}
\max\{0,\rho_{\mu,i}^{-}\}
<
\rho_\mu
<
\min\{1,\rho_{\mu,i}^{+}\},
\label{eq:mu-min-max-condition}
\end{equation}
provided that
\[
\max\{0,\rho_{\mu,i}^{-}\}
<
\min\{1,\rho_{\mu,i}^{+}\}.
\]
The case $\rho_\mu=0$ is included in
\eqref{eq:mu-complete-condition} when
\[
\rho_{\mu,i}^{-}<0<\rho_{\mu,i}^{+}.
\]

We finally compare the two candidate equilibria. Suppose that
$E_2$ exists. Then
\[
\phi(x_2)=\phi_2,
\qquad
0<x_2<x_0.
\]
Since $\phi(0)=0$, $0<\phi_1<\phi_2$, and $\phi(x)$ is
continuous on $[0,x_2]$, there exists
$\widetilde{x}_1\in(0,x_2)$ such that
\[
\phi(\widetilde{x}_1)=\phi_1.
\]
By \eqref{eq:v-phi-relation},
$v(\widetilde{x}_1)=v_1$. Therefore,
$\widetilde{x}_1$ satisfies
\eqref{eq:xi-functional-response} for $i=1$.
Since this equation has a unique positive solution,
$\widetilde{x}_1=x_1$. Hence,
\[
0<x_1<x_2<x_0,
\]
so $E_1$ also exists. Therefore,
\begin{equation}
E_2\text{ exists}
\quad\Longrightarrow\quad
E_1\text{ exists}.
\label{eq:E2-implies-E1}
\end{equation}

Consequently, for each admissible candidate root $\phi_i$, the
conditions \eqref{eq:mu-complete-condition} and $v(x_i)=v_i$ are
necessary and sufficient for $E_i$ to exist, under the prey-persistence
condition $x_0>0$. If two distinct admissible roots are present,
\eqref{eq:E2-implies-E1} shows that $E_2$ cannot exist without $E_1$.
Together with the cases $x_0\leq0$ and $x_0>0$ above, this gives the
classification in Table~\ref{tab:ex-of-eq}.

\subsection{Behavioural classification and data fitting}
\label{app:behavioural-classification}

The source experiments measured drift responses and swimming distance.
In our model, these responses are used to classify individuals as bold
or shy. Prey individuals that drifted within 10 minutes are classified
as bold, whereas predators whose total distance travelled exceeded the
control-group median are classified as bold. The remaining individuals
are classified as shy.

\subsubsection{Prey drift response and fitting of \(\gamma_1(u)\)}

To estimate \(\gamma_1(u)\), we use the aqueous copper exposure data
from \citet{stitt2006evaluation}. The four reported copper
concentrations were \(C<1\ \mu\mathrm{g/L}\),
\(C=3\ \mu\mathrm{g/L}\), \(C=78\ \mu\mathrm{g/L}\), and
\(C=229\ \mu\mathrm{g/L}\). For numerical fitting, the concentration
below the detection limit is represented by
\(C=0.5\ \mu\mathrm{g/L}\).

Following the drift assay, individuals that released their grip on the
substrate and drifted downstream within 10 minutes are classified as
bold for modelling purposes, whereas those that remained attached are
classified as shy. Thus, the fitted proportion of drifting individuals
is used as an operational measure of \(\gamma_1(u)\).

For the numerical fit, we use the 10-minute drift proportions
0.07, 0.23, 0.43, and 0.46
at \(C=0.5\), \(3\), \(78\), and \(229\ \mu\mathrm{g/L}\), respectively.

To assign a body burden to each exposure level, we use the uptake
coefficient \(a_1=0.55\ \mathrm{L/(g\cdot day)}\) from
\citet{huang2015impact}. As a modelling approximation, we use the
10-minute assay period as the uptake interval and set
\(t=10/(60\cdot24)=1/144\) day. The corresponding prey body burden is
estimated by
\begin{equation}
u=a_1Ct
=0.55C\left(\frac{1}{144}\right)
=3.8\times10^{-3}C\ \mu\mathrm{g}/\mathrm{g}.
\end{equation}
This conversion is used as a model proxy and should not be interpreted
as a directly measured body burden.
The four exposure concentrations therefore correspond to
\[
u=1.90\times10^{-3},\quad
1.14\times10^{-2},\quad
0.296,\quad
0.870\ \mu\mathrm{g}/\mathrm{g},
\]
respectively. These estimated body burdens and the four 10-minute
drift proportions listed above are used to fit \(\gamma_1(u)\), as shown in
Figure~\ref{fig:gamma_fitting}.

\subsubsection{Predator swimming response and fitting of \(\gamma_2(v)\)}

To estimate \(\gamma_2(v)\), we use the data from Experiment 4 of
\citet{Glazer2021gammav}. Zebrafish were developmentally exposed to
\(C=0\), \(5\), \(10\), or \(30\) nM MeHg, and their adult swimming
behaviour was later measured in a five-minute novel-tank test. We use
the median total distance travelled in the control group as a fixed
classification threshold. In our model, fish above this threshold are classified as
bold, whereas fish at or below the threshold
are classified as shy. Thus,
\begin{equation}
\gamma_2
=
\frac{\#\{\mathrm{TDT}>\text{control median}\}}
{\#\{\text{total fish}\}},
\end{equation}
where \(\mathrm{TDT}\) denotes total distance travelled. The resulting
bold proportions are \(30/61\), \(20/72\), \(9/54\), and \(4/12\).

The experiment did not report adult body burdens. We therefore map the
developmental exposure concentration \(C\) to the model variable \(v\)
using
\begin{equation}
v=qC.
\label{eq:predator-proxy-scaling}
\end{equation}
The scaling coefficient \(q\) is calculated using
\(1\ \mathrm{nM\ MeHg}\approx0.216\ \mu\mathrm{g/L}\), the uptake
coefficient \(a_2=0.1733\ \mathrm{L/(g\cdot day)}\) from
\citet{huang2015impact}, and a five-minute reference interval:
\begin{equation}
\begin{aligned}
q
&=
a_2
\left(0.216\ \mu\mathrm{g/L\ per\ nM}\right)
\left(\frac{5}{60\cdot24}\ \mathrm{day}\right)
\approx
1.3\times10^{-4}
\ \mu\mathrm{g}\,\mathrm{g}^{-1}\mathrm{nM}^{-1}.
\end{aligned}
\end{equation}
The five-minute interval is used here only to set the scale of \(v\).
For \(C=0,5,10,30\) nM, this gives
\[
v=0,\quad
6.5\times10^{-4},\quad
1.3\times10^{-3},\quad
3.9\times10^{-3}
\ \mu\mathrm{g}/\mathrm{g},
\]
respectively. These scaled values and the four
observed bold proportions are used to fit \(\gamma_2(v)\), as shown in
Figure~\ref{fig:gamma_fitting}.



\begin{figure}[htbp]
    \centering
    \begin{minipage}{0.48\textwidth}
        \centering
        \includegraphics[width=\linewidth]{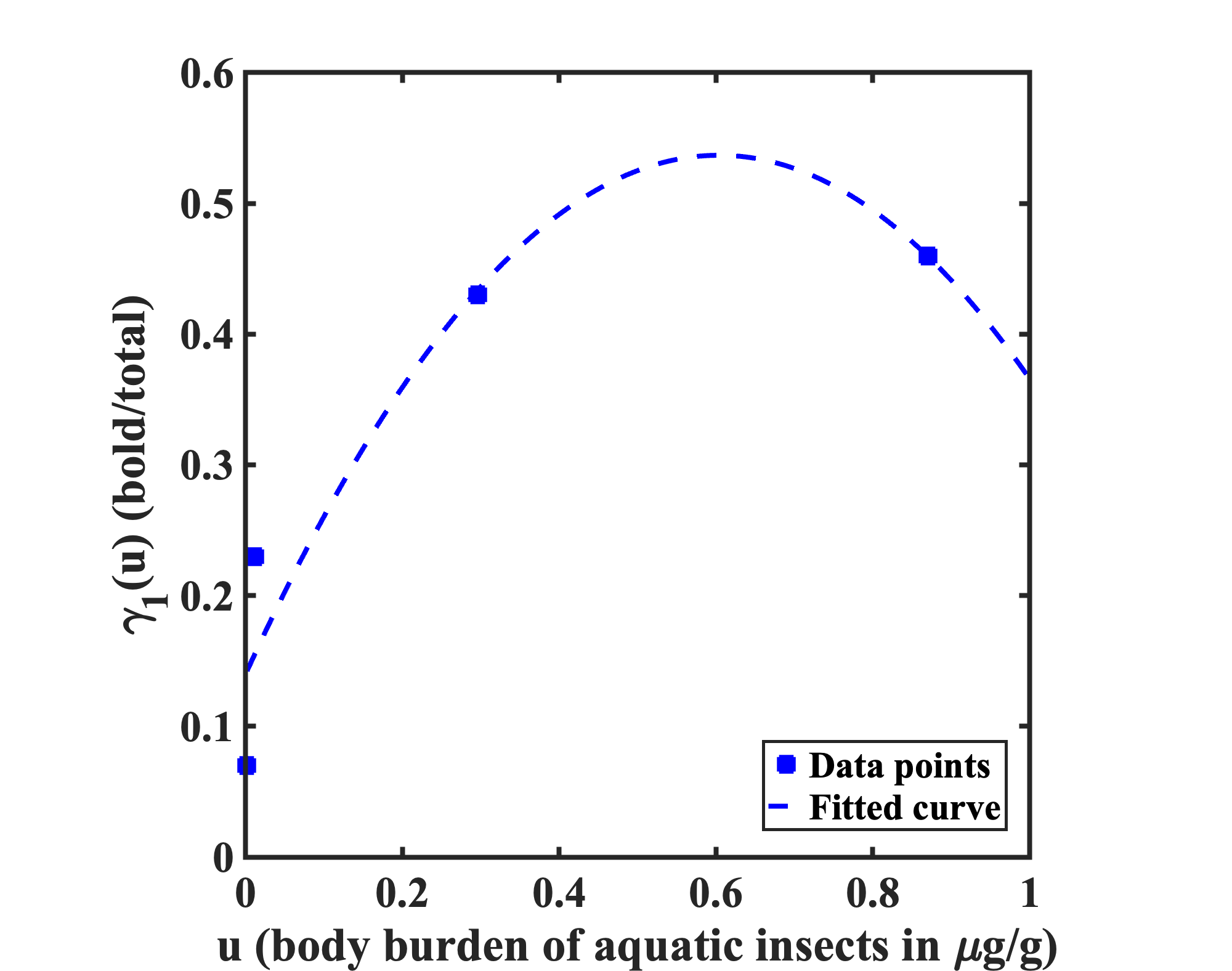}
        \caption*{(a)}
    \end{minipage}
    \hfill
    \begin{minipage}{0.48\textwidth}
        \centering
        \includegraphics[width=\linewidth]{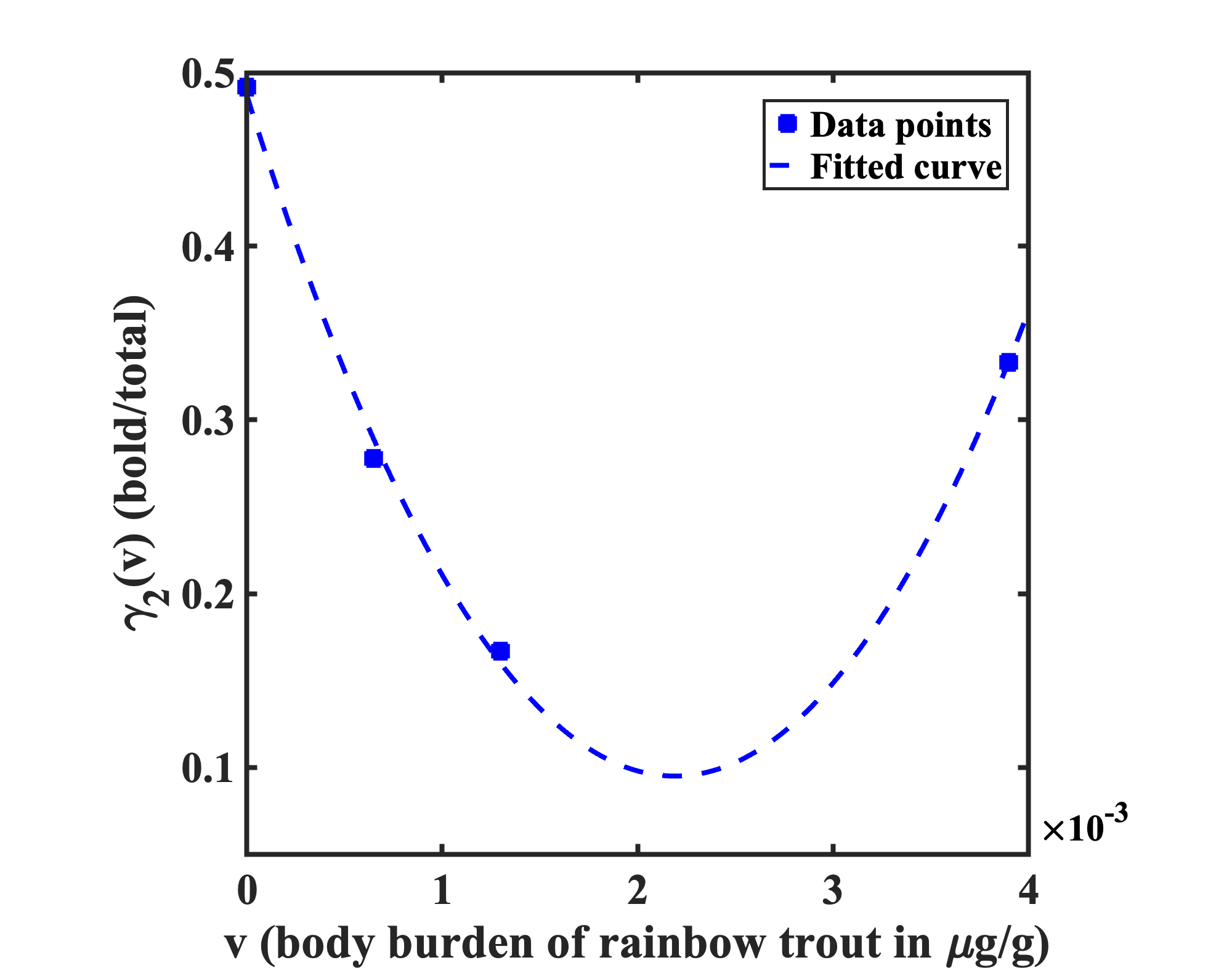}
        \caption*{(b)}
    \end{minipage}

    \caption{Fitted toxin-dependent proportions of individuals classified
    as bold based on the measured behavioural responses.
    (a) Fitted proportion of drifting prey, $\gamma_1(u)$, based on
    \citet{stitt2006evaluation}.
    (b) Fitted proportion of predators above the control-median
    swimming-distance threshold, $\gamma_2(v)$, based on Experiment~4
    of \citet{Glazer2021gammav}.}
    \label{fig:gamma_fitting}
\end{figure}

\section{Notation}
\label{sec: Notation}
The quantities used in Table~\ref{tab:ex-of-eq} are defined as follows:
\begin{flalign*}
T_1^* &:=
\frac{\beta_1\sigma_2\left(\beta_1c+k_2\beta_2+2\beta_2m_2\right)
-\beta_1\sigma_2\sqrt{4\beta_2\left(\beta_1c+\beta_2m_2\right)
\left(m_2+k_2\right)}}{\left(\beta_1c-k_2\beta_2\right)^2}, &&\\[5pt]
T_2^* &:=
\frac{(\beta_1-m_2)\sigma_2}
{\beta_1\beta_2+\beta_1c+k_2\beta_2+k_2c}, &&\\[5pt]
T_3^* &:=
\frac{\beta_1\sigma_2}
{\beta_1c+k_2\beta_2+2\beta_1\beta_2}, &&\\[5pt]
\gamma_1 &:=
\gamma_1\left(\frac{T}{\widehat{\alpha}_2}\right), &&\\[5pt]
v_i &:=
\frac{cT+\beta_2T\phi_i}{\sigma_2},
\qquad
\gamma_{2,i}:=
\gamma_2\left(\frac{v_i}{\widehat{\beta}_2}\right),
\qquad i\in\{1,2\}, &&\\[5pt]
x_0 &:=
\bar{\alpha}\left[1+\rho_\alpha(2\gamma_1-1)\right]
\frac{1-T}{k_1T+m_1}-1, &&\\[5pt]
\rho_\alpha^*(T) &:=
\frac{\dfrac{k_1T+m_1}{\bar{\alpha}(1-T)}-1}{2\gamma_1-1},
\qquad \gamma_1\neq\frac12, &&\\[5pt]
\phi_{1,2}(T) &:=
\frac{\beta_1(\sigma_2-cT)-k_2\beta_2T\mp\sqrt{\Delta(T)}}
{2\beta_1\beta_2T}, &&\\[5pt]
\Delta(T) &:=
\left[\beta_1(\sigma_2-cT)-k_2\beta_2T\right]^2
-4\beta_1\beta_2T\left(k_2cT+m_2\sigma_2\right), &&\\[5pt]
C_i &:=
(1-\phi_i)x_0^2+(1-2\phi_i)\bar{\eta}x_0
-\phi_i\bar{\eta}^2,
\qquad i\in\{1,2\}, &&\\[5pt]
\Delta_{\mu,i} &:=
x_0^2\left[\bar{\eta}^2(2\gamma_{2,i}-1)^2
+4(1-\phi_i)C_i\right],
\qquad i\in\{1,2\}, &&\\[5pt]
\rho_{\mu,i}^{\pm} &:=
\frac{\bar{\eta}(2\gamma_{2,i}-1)
\pm\sqrt{\bar{\eta}^2(2\gamma_{2,i}-1)^2
+4(1-\phi_i)C_i}}
{2(1-\phi_i)x_0},
\qquad i\in\{1,2\}, &&\\[5pt]
A_i &:=
\eta_1+(\eta_2-\eta_1)\gamma_{2,i}
-\phi_i(\eta_1+\eta_2),
\qquad i\in\{1,2\}, &&\\[5pt]
x_i &:=
\frac{-A_i+\sqrt{A_i^2+4\phi_i(1-\phi_i)\eta_1\eta_2}}
{2(1-\phi_i)},
\qquad i\in\{1,2\}. &&
\end{flalign*}

\section*{Acknowledgments}
The research of Hao Wang was partially supported by the Natural Sciences and Engineering Research Council of Canada (Individual Discovery Grant RGPIN-2025-05734) and the Canada Research Chairs program (Tier 1 Canada
Research Chair Award).


\end{document}